\documentclass[a4paper,onecolumn,11pt,accepted=2024-09-19]{quantumarticle}
\pdfoutput=1
\usepackage[utf8]{inputenc}
\usepackage[english]{babel}
\usepackage[T1]{fontenc}
\usepackage{amsmath}
\usepackage{hyperref}
\usepackage[numbers]{natbib}

\usepackage{tikz}
\usepackage{amssymb}
\usepackage{tikz}
\usepackage{mathrsfs}
\usepackage{graphicx}
\usepackage{comment}
\usepackage{dsfont}
\usepackage{amsmath}
\usepackage{amsfonts}
\usepackage{amssymb}
\usepackage{braket}
\usepackage{amsthm}
\usepackage{caption}
\usepackage{lipsum}
\usepackage{subcaption}
\usepackage{mwe}
\usepackage{enumerate}
\usepackage{sidecap}
\usepackage{appendix}
\usepackage{multirow}
\usepackage{array}
\usepackage{physics}
\usepackage{mathtools}
\usepackage{enumitem}

\newtheorem{theorem}{Theorem}
\newtheorem{lemma}[theorem]{Lemma}
\newtheorem{proposition}[theorem]{Proposition}
\newtheorem{corollary}[theorem]{Corollary}
\newtheorem{definition}[theorem]{Definition}
\newtheorem{remark}{Remark}

\newcommand{\state}{\mathcal{S}(A)}
\newcommand{\F}{\bar{S}}
\newcommand{\mbI}{\mathbb{I}}
\newcommand{\mS}{\mathcal{S}}
\newcommand{\mfF}{\mathfrak{F}}

\DeclareMathOperator*{\argmin}{arg\,min}

\DeclareMathOperator{\supp}{supp}
\DeclareMathOperator{\sinc}{sinc}

\renewcommand{\tilde}[1]{\widetilde{#1}} 
\renewcommand{\hat}[1]{\widehat{#1}}

\def \Trm {\text{Tr}} 

\def \d {\mathrm{d}}

\begin{document}

\title{Mixed-state additivity properties of magic monotones based on quantum relative entropies for single-qubit states and beyond}

\author{Roberto Rubboli}
\email{roberto.rubboli@u.nus.edu}
\affiliation{Centre for Quantum Technologies, National University of Singapore, Singapore 117543, Singapore}
\orcid{0000-0002-7914-671X}
\author{Ryuji Takagi}
\orcid{0000-0003-3837-8159}
\affiliation{Department of Basic Science, The University of Tokyo, Tokyo 153-8902, Japan}
\affiliation{Nanyang Quantum Hub, School of Physical and Mathematical Sciences,
Nanyang Technological University, 637371, Singapore}
\author{Marco Tomamichel}
\affiliation{Centre for Quantum Technologies, National University of Singapore, Singapore 117543, Singapore}
\affiliation{Department of Electrical and Computer Engineering,
National University of Singapore, Singapore 117583, Singapore}
\orcid{0000-0001-5410-3329}
\maketitle

\begin{abstract}
  We prove that the stabilizer fidelity is multiplicative for the tensor product of an arbitrary number of single-qubit states.  We also show that the relative entropy of magic becomes additive if all the single-qubit states but one belong to a symmetry axis of the stabilizer octahedron. We extend the latter results to include all the $\alpha$-$z$ R\'enyi relative entropy of magic. This allows us to identify a continuous set of magic monotones that are additive for single-qubit states.  We also show that all the monotones mentioned above are additive for several standard two and three-qubit states subject to depolarizing noise. Finally, we obtain closed-form expressions for several states and tighter lower bounds for the overhead of probabilistic one-shot magic state distillation.
\end{abstract}

\section{Introduction}
The resource theory of magic provides a framework to quantify the advantage of quantum computation over classical computation~\cite{veitch2014resource}. In this setting, free resources lead to efficiently classically simulable computation, while resource states or `magic states' unlock the quantum advantage over the classical counterpart allowing for universal quantum computation. The core idea is that by injecting magic states into the circuit, using only stabilizer operations, it is possible to implement non-stabilizer gates and achieve universal quantum computation. According to the Gottesman-Knill theorem~\cite{gottesman1998theory}, stabilizer operations can be efficiently classically simulated while the simulation runtime scales exponentially with the number of injected magic states into the circuit~\cite{aaronson2004improved}. Magic monotones quantify the amount of magic resources contained in a quantum state. It turns out that the qudit magic theory for odd-dimensional Hilbert spaces is structurally simple~\cite{veitch2014resource,gross2006hudson,veitch2012negative,mari2012positive}. This is because the Wigner function is well-behaved. In this case, many monotones have been defined and multiplicativity is a common occurrence. Some prominent examples include the mana introduced in the seminal work~\cite{veitch2014resource} and the min and max-thauma introduced in~\cite{wang2020efficiently}. However, despite many attempts, alternative ways of defining a well-behaved Wigner function for qubits suffer from drawbacks, and any attempt to define additive monotones for all qubit states proved to be highly challenging~\cite{seddon2021quantifying}. 

The relative entropy measure is often a common choice in resource theories since, once regularized, it provides upper bounds on asymptotic conversion rates between states in any convex resource theory~\cite{horodecki2013quantumness}. In magic theory, the so-called relative entropy of magic was first introduced in~\cite{veitch2014resource}. However, although it was introduced almost a decade ago, its properties are known only for the few cases where it relates to other resource monotones~\cite{seddon2021quantifying,Takagi2022one-shot}. Another popular choice is the fidelity measure, which quantifies the resource of a state in terms of the infidelity from the set of free states. In the context of magic theory, it was introduced in~\cite{veitch2014resource} and it is known as stabilizer fidelity. Finally, we mention the generalized robustness of magic, which quantifies how much mixing with a state can take place before a magic state becomes a stabilizer state~\cite{Liu2019oneshot,seddon2021quantifying}. The latter monotone benchmarks several proposed classical simulators (see~\cite{seddon2021quantifying} for more details). Moreover, for single-qubit states, the generalized robustness of magic is equal to the dyadic negativity and the mixed-state extent~\cite[Theorem 3]{seddon2021quantifying}. The relative entropy of magic, the stabilizer fidelity, and the generalized robustness of magic are instances of magic monotones based on a quantum relative entropy and can be defined for any resource theory.  A fundamental problem is whether the magic monotones based on a quantum relative entropy are additive under tensor products, and if not, under which conditions they become additive. However, very little is known about the additivity properties of these monotones. Besides its theoretical interest, the latter property is practically useful since it allows us to reduce the computation of the magic resource of the `whole' to just the sum of one of the `parts'. The computation of these monotones becomes soon unfeasible for multi-qubit states as the number of qubits grows. This is because the number of pure stabilizer states scales super-exponentially with the number of qubits. For numerical experiments, we are limited to just five-qubit states~\cite{seddon2021quantifying,howard2017application}. The additivity property allows us to overcome this problem for tensor products of low-dimensional multi-qubit states. In this case, the total magic resource just equals the sum of the one of each state, which is computable.

We now summarize the main results of this work. 
\begin{itemize}
\item In Section~\ref{multiplicativity fidelity}, we prove that the stabilizer fidelity is multiplicative for tensor products of any number of single-qubit states. We then extend this result to include a continuous set of magic monotones, namely all the $\alpha$-$z$ R\'enyi relative entropies of magic in the range $|(1-\alpha)/z|=1$. In this way, we also independently recover the multiplicativity of the generalized robustness of magic for all single-qubit states established in~\cite{seddon2021quantifying}. 

\item In Section~\ref{relative entropy additivity}, we prove that the relative entropy of magic is additive for single-qubit states if all of them but one commute with their optimizer, i.e., the closest stabilizer state.  Moreover, in Section~\ref{specific states}, we show that the latter class includes the $T$,$H$, and $F$ states subject to depolarizing noise. These are the states involved in the standard magic state distillation scenario where any initial state is first projected onto a symmetry axis of the octahedron by applying the stabilizer twirling, from which one aims to distill a better pure magic state~\cite{campbell2012magic,campbell2010bound}. We extend this result to all the $\alpha$-$z$ R\'enyi relative entropies of magic. Finally, we also show that the relative entropy of magic is not additive for general two multi-qubit states by providing a counterexample.

\item In Section~\ref{multi-qubit additivity}, we prove that  all the $\alpha$-$z$ R\'enyi relative entropies of magic are additive for two and three-qubit states that belong to a specific class. In Section~\ref{specific states}, we show that the latter class includes the Toffoli, Hoggar, and $CS$ states subject to global depolarizing noise.

\item In Section~\ref{Magic state transformation} we show that the $\alpha$-$z$ R\'enyi relative entropies of magic provide tighter probabilistic lower bounds near the deterministic region for the overhead of magic state distillation. The additivity results obtained in the previous sections extend these results to the case where multiple copies are involved in the process. Moreover, they extend these results to scenarios involving catalysts that assist the transformation and could potentially provide a better overhead. 

\item In Section~\ref{closed forms}, we show that two additive quantities for single-qubit states, namely the generalized robustness and the stabilizer fidelity, have a closed for all single-qubit states. In Section~\ref{specific states} we give a closed-form for several two and three-qubit states belonging to the additivity classes for all the $\alpha$-$z$ R\'enyi relative entropies of magic. Explicitly, we consider $T,H,F$, Toffoli, Hoggar, and $CS$ magic states subject to global depolarizing noise. 

\item In Section~\ref{counterexample}, we show that for general odd-dimensional states, no monotone based on a quantum relative entropy is additive. 

\item Finally, in Section~\ref{Further directions} we discuss a method that potentially could be useful to prove the additivity of the $\alpha$-$z$ R\'enyi relative entropies of magic in the range $|(1-\alpha)/z|=1$ for any number of any two and three-qubit states.

\end{itemize}

\section{Magic states and magic monotones}
\label{Magic monotones}
We denote with $\mathcal{P}(A)$ the set of positive operators on a Hilbert space $A$. Moreover, we denote with $\state$ the set of quantum states, i.e., the subset of $\mathcal{P}(A)$ with unit trace. A resource theory is defined by a set $\mathcal{F}$ of free states and a set $\mathcal{O}_{\mathcal{F}}$ of free operations with the property that map free states into free states~\cite{Gour,Brandao2}. In the resource theory of magic, the free states are the set of stabilizer states $\mathcal{F}=\F$~\cite{veitch2014resource}. The stabilizer states are defined as follows. The $n$-qubit Clifford group $\mathcal{C}_{n}$ is the group of $n$-qubit unitaries generated by the single-qubit gates H, S, and the CNOT two-qubit gate. Here, H is the Hadamed gate, S is the phase gate and CNOT is the controlled-not gate. The pure $n$-qubit stabilizer states are the ones generated by the Clifford group acting on $\ket{0^n}$, i.e., all the states $S_n=\{ U\ket{0^n} : U \in \mathcal{C}_{n}\}$. The set of all (mixed) stabilizer states is formed by all convex combinations (convex hull) of pure stabilizer states, $ \bar{S}_n :=  \text{conv}\{\ketbra{\psi}{\psi} : \ket{\psi} \in S_n \} $. We omit the subscript $n$ if the geometry is clear from the context. We say that a state $\rho$ is magic if it is not a stabilizer state. 

For single-qubit states, the set of stabilizer states has a simple geometric characterization inside the Bloch sphere.  We say that a pure state $\ket{\psi}$ is stabilized by an operator $P$ if $P\ket{\psi} = \ket{\psi}$.  The single-qubit pure stabilizer states are the states that are stabilized by the Pauli operators as well as their products with $\pm \mathds{1}$. These are the states $\ket{0}$, $\ket{1}$, $\ket{+}:=(\ket{0}+\ket{1})/\sqrt{2}$, $\ket{-}:=(\ket{0}-\ket{1})/\sqrt{2}$, $\ket{+i}:=(\ket{0}+i\ket{1})/\sqrt{2}$, $\ket{-i}:=(\ket{0}-i\ket{1})/\sqrt{2}$.  The set of all the stabilizer states in the Bloch sphere forms an octahedron whose vertexes are the pure stabilizer states. In this work, we mainly focus on multi-qubit systems. 
The resource theory of magic can also be formulated for higher-dimensional qudit systems. We refer to~\cite{veitch2014resource,campbell2012magic} for a formal definition. 
In the following, we sometimes consider some standard magic states subject to depolarizing noise. The depolarizing channel is defined by the map
\begin{equation}
\Delta_p: \rho \rightarrow p \rho + (1-p) \frac{\mathds{1}}{d},
\end{equation}
where $p \in [0,1]$.

Resource monotones quantify the resource of a quantum state.
We call a function $\mathfrak{R} : \mathcal{S}(A) \rightarrow [0, + \infty]$ a resource monotone if it does not increase under free operations, i.e., if $\mathfrak{R}(\rho) \geq \mathfrak{R}(\mathcal{E}(\rho))$ for any state $\rho$ and any free operation $\mathcal{E} \in \mathcal{O}_\mathcal{F}$. In the resource theory of magic, resource monotones are called magic monotones. In this work, we will mainly be concerned with the magic monotones based on the $\alpha$-$z$ R\'enyi relative entropies. This allows us to address in a compact way several magic monotones already introduced in the literature.  Although there are several ways to define free operations in the resource theory of magic~\cite{veitch2014resource,seddon2019quantifying,seddon2021quantifying}, in this work we do not specify any set of free operations since the $\alpha$-$z$ R\'enyi relative entropies of magic are monotone for any choice of free operations. We say that $\mathfrak{R}$ is \textit{tensor sub-additive} (or just sub-additive) for the states $\rho_1$ and $\rho_2$ if $\mathfrak{R}(\rho_1 \otimes \rho_2) \leq \mathfrak{R}(\rho_1) +  \mathfrak{R}(\rho_2)$.  Moreover, we say that $\mathfrak{R}$ is \textit{tensor additive} (or just additive) for the states $\rho_1$ and $\rho_2$ if $\mathfrak{R}(\rho_1 \otimes \rho_2) = \mathfrak{R}(\rho_1) +  \mathfrak{R}(\rho_2)$. 

We now introduce the $\alpha$-$z$ R\'enyi relative entropies. Let $ \alpha \in (0,1) \cup(1,\infty), \; z>0$, $\rho \in \state$ and $ \sigma \in \mathcal{P}(A)$. Then the $\alpha$-$z$ \textit{R\'enyi relative entropy} of $\sigma$ with $\rho$ is defined as~\cite{audenaert2015alpha,zhang2020wigner}
\begin{equation}
\label{definition}
D_{\alpha,z}(\rho \| \sigma):=
\begin{cases}
\frac{1}{\alpha-1}\log{\text{Tr}\Big[\big(\rho^\frac{\alpha}{2z}\sigma^\frac{1-\alpha}{z}\rho^\frac{\alpha}{2z}\big)^z\Big]} & \text{if}\; (\alpha<1 \wedge \rho \not \perp \sigma) \vee \rho \ll \sigma \\
 +\infty & \text{else}
\end{cases} \,.
\end{equation}
In the following, we denote 
\begin{equation}
Q_{\alpha,z}(\rho \| \sigma):=\exp \big( (\alpha-1)D_{\alpha,z}(\rho \| \sigma) \big)\,.
\end{equation}
In the limit points of the ranges of the parameters, we define the $\alpha$-$z$ R\'enyi relative entropy by taking the corresponding pointwise limits.
In particular, throughout the text, we will often be concerned with the pointwise limits
\begin{align}
\label{limits}
&D_{\min}(\rho \| \sigma) = \lim \limits_{\alpha \rightarrow 0} D_{\alpha,1-\alpha}(\rho \| \sigma) \,, \\
&D(\rho \| \sigma) = \lim \limits_{\alpha \rightarrow 1} D_{\alpha,\alpha}(\rho \| \sigma)\,, \\
&D_{\max}(\rho \| \sigma) = \lim \limits_{\alpha \rightarrow \infty} D_{\alpha,\alpha-1}(\rho \| \sigma)   \,,
\end{align}
where 
\begin{align}
\label{Dmin}
& D_{\min}(\rho \| \sigma) := -\log{\Tr(\Pi(\rho)\sigma)} \, , \\
&D(\rho \| \sigma) := \Tr(\rho (\log{\rho} - \log{\sigma}))\,, \,  \text{and}\\
&D_{\max}(\rho \| \sigma) := \inf \{ \lambda \in \mathbb{R} : \rho \leq 2^{\lambda}\sigma \} \, ,
\end{align}
are the \textit{min-relative entropy}~\cite{renner2008security,Datta_rob2}, the Umegaki relative entropy, and the  \textit{max-relative entropy}~\cite{tomamichel2015quantum,Datta_rob2,renner2008security}, respectively. In Appendix~\ref{convergence monotone}, we prove the first and the third limit of~\eqref{limits}. Note that the second limit in~\eqref{limits} for the Umegaki relative entropy holds for any $z > 0$~\cite{lin2015investigating}. Here, we denoted with $\Pi(\rho)$ the projector onto the support of $\rho$.

In~\cite[Theorem 1.2]{zhang2020wigner}, the author proved that the $\alpha$-$z$ R\'enyi relative entropy satisfies the data-processing inequality (DPI) for the following range of parameters
\begin{enumerate}
\item  $\quad 0<\alpha<1 \; \; and \; \; z\geq \max\{\alpha,1-\alpha\},$ 
\item $\quad 1<\alpha \leq 2 \;\;  and \;\;  \frac{\alpha}{2} \leq z \leq \alpha,\; and $
\item $\quad 2 \leq \alpha<\infty \;\;  and \;\;  \alpha-1 \leq z \leq \alpha\, . $
\end{enumerate}
We denote with $\mathcal{D}$ the set of above values of the parameters $(\alpha,z)$ for which the $D_{\alpha,z}$ satisfies the DPI.
Because of the limits in~\eqref{limits}, the $\alpha$-$z$ R\'enyi relative entropy also satisfies the DPI on the line $\alpha=1$. We include this line in the region $\mathcal{D}$. 
In Fig.~\ref{fig: alpha-z}, we represent in blue the region for which the $\alpha$-$z$ R\'enyi relative entropies satisfy the DPI.

\begin{figure}
\centering
{\includegraphics[width=1\textwidth]{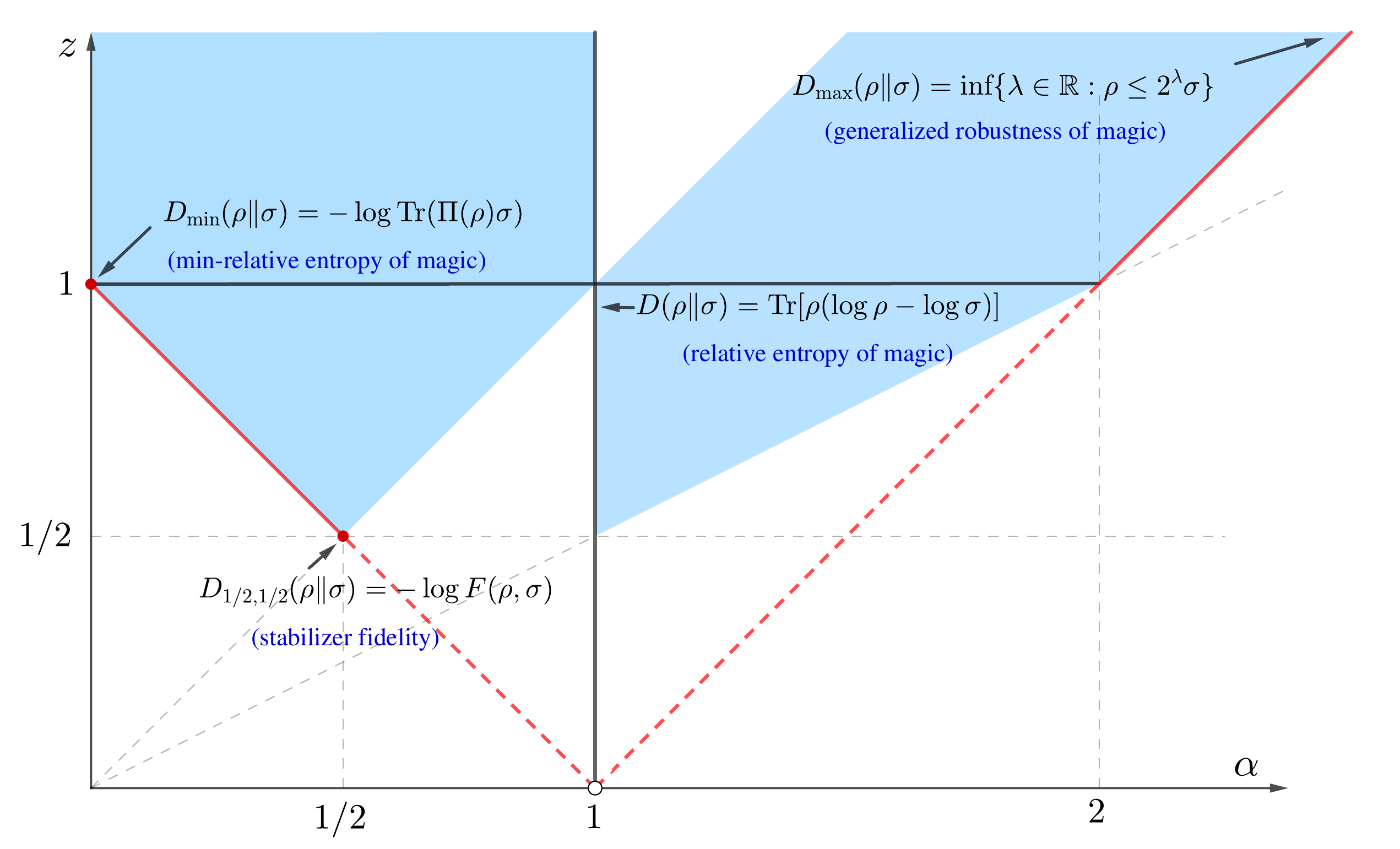}}
\caption{$\alpha$-$z$ plane for the $\alpha$-$z$ R\'enyi relative entropies. The vertical line $\alpha=1$ for $z > 0$ corresponds to the Umegaki relative entropy. The blue region (together with the Umegaki relative entropy line) is where the data-processing inequality holds. We show in red the special line $|(1-\alpha)/z|=1$ located at the boundary of the DPI region. We write in blue the related magic monotones associated with the respective points. In Section~\ref{Magic state transformation} we use the additivity of the monotones in a neighborhood of $\alpha=0, z=1$  to obtain tighter one-shot bounds for magic state distillation.}
\label{fig: alpha-z}
\end{figure}

We define magic monotones for $(\alpha,z)\in\mathcal{D}$ as
\begin{align}
\label{problem}
&\mathfrak{D}_{\alpha,z} (\rho) := \inf_{\sigma \in \F} D_{\alpha,z}(\rho \| \sigma) \,.
\end{align}
Analogously, we denote 
\begin{equation}
\mathcal{Q}_{\alpha,z}(\rho):=\exp \big( (\alpha-1)\mathfrak{D}_{\alpha,z}(\rho) \big)\,.
\end{equation}
In the following, we call these monotones R\'enyi relative entropies of magic.
In particular, we have
\begin{equation}
\label{definitions}
\mathfrak{D}_{\min}(\rho) :=   \inf_{\sigma \in \F} D_{\min}(\rho \| \sigma) , \quad \mathfrak{D}(\rho) :=   \inf_{\sigma \in \F} D(\rho \| \sigma) , \quad \mathfrak{D}_{\max}(\rho):= \inf_{\sigma \in \F} D_{\max}(\rho \| \sigma) \,.
\end{equation}
Note that since the $\alpha$-$z$ R\'enyi relative entropies are lower semicontinuous, we can replace the above infimum with a minimum~\cite{rubboli2024new}.
In Appendix~\ref{convergence monotone} we prove that we can recover the latter monotones as pointwise limits of the monotone~\eqref{problem}. Namely, we have
\begin{equation}
\label{limits monotone}
\mathfrak{D}_{\min}(\rho) = \lim \limits_{\alpha \rightarrow 0}\mathfrak{D}_{\alpha,1-\alpha}(\rho) , \quad \mathfrak{D}(\rho) = \lim \limits_{\alpha \rightarrow 1}\mathfrak{D}_{\alpha,\alpha}(\rho) , \quad \mathfrak{D}_{\max}(\rho) = \lim \limits_{\alpha \rightarrow \infty}\mathfrak{D}_{\alpha,\alpha-1}(\rho) \,.
\end{equation}
The $\mathfrak{D}_{\alpha,z}$ do not increase under any set of free operations that map stabilizer states into stabilizer states, i.e., they are indeed magic monotones. The latter property is a straightforward consequence of the fact that the underlying $D_{\alpha,z}$ satisfies the data-processing inequality. Moreover, it is easy to prove that the $\mathfrak{D}_{\alpha,z}$ are sub-additive for general states. 

The $\alpha$-$z$ R\'enyi relative entropies encompass many well-known quantities. The relative entropy of magic was first introduced in the seminal work~\cite{veitch2014resource} and was the first monotone based on a quantum relative entropy to be proposed.  The relative entropy of magic is equal to $\mathfrak{D}$ defined in~\eqref{definitions}. We remark that our definition differs from the relative entropy of magic introduced in~\cite{bu2023quantum}. Explicitly, we consider the case where the relative entropy is minimized over all stabilizer states. However, in~\cite{bu2023quantum}, the authors consider the minimization over the set of minimal stabilizer projection states. The stabilizer fidelity was first introduced in~\cite{bravyi2019simulation} and is defined as
\begin{equation}
\mathfrak{F}(\rho) = \sup \limits_{\sigma \in \F}F(\rho, \sigma)
\,,
\end{equation}
where $F(\rho, \sigma) :=  ( \Trm|\sqrt{\rho}\sqrt{\sigma}| )^2$ is the Uhlmann's fidelity. We have the relationship $\mathfrak{D}_{\frac{1}{2},\frac{1}{2}}(\rho)  =-\log{\mathfrak{F}(\rho)}$. The generalized robustness of magic was later introduced in~\cite{Liu2019oneshot,seddon2021quantifying} and it is defined as
\begin{equation}
\Lambda^{+}(\rho) := \inf \left\{ s\geq 1: \exists \omega \in \mathcal{S}(A) \,\, \text{s.t} \,\, \frac{1}{s}\left(\rho + (s-1)\omega\right) \in \F  \right\}\, .
\end{equation}
The generalized robustness of magic is connected to the monotone based on the max-relative entropy through the relation $\mathfrak{D}_{\max}(\rho)=\log{(\Lambda^+(\rho))}$. In Fig.~\ref{fig: alpha-z} we show where the latter monotones are located in the $\alpha$-$z$ plane. We show in red the boundary of the DPI region $|(1-\alpha)/z|=1$. Notably, the latter line contains the stabilizer fidelity, the min, and the max-relative entropy of magic. As we discuss below, these monotones on the boundary satisfy stronger additivity properties than the one located at the interior of the DPI region.

 For a monotone to be called an entanglement measure, further properties such as faithfulness, convexity, or strong monotonicity are usually required~\cite{horodecki2009quantum,vedral1998entanglement,vedral1997quantifying}. As we discuss in Appendix~\ref{distillation lower bound}, the relative entropy of magic, the stabilizer fidelity, and the generalized robustness of magic satisfy these properties. Moreover, for any $(\alpha,z) \in \mathcal{D}$ it is possible to construct quantities with these desired properties. This fact is very general and holds for a wide class of resource theories. 
 
In the next sections, we focus on the additivity properties of the magic monotones based on the $\alpha$-$z$ R\'enyi relative entropies. Note that we can restrict our attention only to magic states. Indeed, for free states, additivity readily follows from data processing under partial trace and sub-additivity of the monotones as noted in~\cite[Theorem 7]{veitch2014resource}. We summarize in the table below the main findings.

\setlength{\tabcolsep}{10pt}
\renewcommand{\arraystretch}{1.6}
\begin{table}
\begin{tabular}{l  l  l} 
 &  $|(1-\alpha)/z|=1$ &  $|(1-\alpha)/z|\neq1$ \\  
\hline
\multirow{2}{*}{single-qubit}  &  \multirow{2}{*}{all  (Thm.~\ref{Additivity in the line})} &  not all  (Subsec.~\ref{counterexample qubits})\\  
& & class 1 (Thm.~\ref{theorem relative entropy of magic qubits}) \\
\hline
\multirow{2}{*}{two/three-qubit} & conjecture all (Sec.~\ref{Further directions})  &  not all (Subsec.~\ref{counterexample qubits}) \\  
& class 2  (Thm.~\ref{theorem relative entropy of magic})  &  class 2  (Thm.~\ref{theorem relative entropy of magic})\\
\hline
multi-qubit ($n \geq 12$) &  not all (Sec.~\ref{multi-qubit additivity}) &  not all (Subsec.~\ref{counterexample qubits})\\  
\hline
qutrit &  not all (Sec.~\ref{counterexample})&  not all (Sec.~\ref{counterexample}) \\  
\hline
\end{tabular}
\caption{The table shows the additivity properties of the monotones based on the $\alpha$-$z$ R\'enyi relative entropies as a function of the system dimension. We split the range of the parameters in two sets, namely $|(1-\alpha)/z|=1$ and $|(1-\alpha)/z| \neq 1$. Indeed, they show different additivity properties. For the first range, additivity holds for all single-qubit states (Theorem~\ref{Additivity in the line}) and a specific class of two and three-qubit states (Thereom~\ref{theorem relative entropy of magic}). However, in Section~\ref{Further directions} we conjecture that additivity holds for all states. For multi-qubit states with more than $12$ qubits, additivity does not generally hold (see Section~\ref{multi-qubit additivity}). For qutrit states, additivity does not generally hold (see Section~\ref{counterexample}). The monotones in the range $|(1-\alpha)/z| \neq 1$ are not additive for single-qubit states~(see Subsection~\ref{counterexample qubits}) but only for specific states (Theorem~\ref{theorem relative entropy of magic qubits}). For two and three-qubit states, they are additive for a specific class (Theorem~\ref{theorem relative entropy of magic}). For multi-qubit states ($n \geq 12$) and qutrit states, they are not additive.} 
\label{states}
\end{table}

\section{Additivity of the $\alpha$-$z$ R\'enyi relative entropies of magic in the range $|(1-\alpha)/z|=1$ for all single-qubit states}
\label{multiplicativity fidelity}
In this section, we prove that the stabilizer fidelity is multiplicative for an arbitrary number of tensor products of single-qubit states. Previously, it was known only for pure states~\cite{bravyi2019simulation}. We generalize the result for all the $\alpha$-$z$ R\'enyi relative entropies of magic in the region $|(1-\alpha)/z|=1$. In this way, by taking the limit $\alpha \rightarrow \infty$ and $\alpha \rightarrow 0$, we independently recover the multiplicativity of the generalized robustness of magic and the additivity of the $\mathfrak{D}_{\min}$ established in~\cite{seddon2021quantifying} and~\cite{saxena2022quantifying}, respectively.

The key ingredient for the proof of the stabilizer fidelity is the multiplicativity for all single-qubit states of the linear functional
\begin{equation}
     \mathfrak{F}_0(\rho):= \max \limits_{\sigma \in \F}\Trm(\rho \sigma)\,.
\end{equation}
Note that the above function is simpler than the fidelity due to its linearity. As we show later, the multiplicativity of the above quantity is tightly connected to the one of the stabilizer fidelity. 

We now show that $\mathfrak{F}_0$ is multiplicative for an arbitrary number of single-qubit states. 
\begin{proposition}
\label{multiplicativity F_0}
Let $\{\rho_i\}_{i=1}^L$ be a set of single-qubit states. Then, we have
    \begin{equation}
         \mathfrak{F}_0 \left( \bigotimes_{i=1}^L \rho_i \right) = \prod_{i=1}^L  \mathfrak{F}_0(\rho_i) \,.
    \end{equation}
\end{proposition}
    \begin{proof}
        Let $r_i=(x_i,y_i,z_i)$ be the Bloch vectors of $\rho_i$ for $i=1,2$. 
        Without loss of generality, since $\mathfrak{F}_0$ is invariant under Clifford unitaries, let us assume that $z_i \geq y_i \geq x_i \geq 0$. Using the spectral decomposition, we can decompose any qubit state as a convex combination between a pure state and the identity. Indeed, we have
\begin{equation}
\rho = p \psi_{r} + (1-p)  \psi_{-r} = (2p-1) \psi_{r} + (1-p) \mathds{1}\,,
\end{equation}
where we denoted with $\psi_{r}$ and $\psi_{-r}$ the eigenvectors of $\rho$. Note that, by assumption, $p \geq 1/2$. We, therefore, obtain that
\begin{align}
\mathfrak{F}_0(\rho) &= \max \limits_{\sigma \in \F} \Trm(\rho \sigma) \\
 &= (2p-1) \max \limits_{\sigma \in \F} \Trm(\psi_{r} \sigma) + (1-p) \\
 \label{one state}
&=  a\mathfrak{F}_0(\psi_{r}) + b \,,
\end{align}
where we defined the positive constant $a:= (2p-1)$ and $b :=  (1-p)$. Hence, the optimum is achieved by the state that maximizes the overlap with the pure state in the decomposition with the highest weight. 
We now consider the tensor product of two states. Using the above decomposition for both states, we obtain
\begin{align}
\mathfrak{F}_0(\rho_1 \otimes \rho_2) =& a_1 a_2\Trm(\tau \psi_{r_1} \otimes \psi_{r_2}) + a_1b_2\Trm(\Trm_2(\tau)\psi_{r_1}) + b_1 a_2 \Trm(\Trm_1(\tau)\psi_{r_2}) + b_1 b_2 \\
\leq & a_1 a_2 \mathfrak{F}_0(\psi_{r_1})\mathfrak{F}_0(\psi_{r_2}) + a_1 b_2\mathfrak{F}_0(\psi_{r_1}) + b_1 a_2 \mathfrak{F}_0(\psi_{r_2}) + b_1 b_2 \\
=& (a_1 \mathfrak{F}_0(\psi_{r_1}) + b_1)(a_2 \mathfrak{F}_0(\psi_{r_2}) + b_2) \\
= & \mathfrak{F}_0(\rho_1)\mathfrak{F}_0(\rho_2) \,,
\end{align}
where we used $\rho_{1,2}=a_{1,2}\psi_{r_{1,2}}+b_{1,2}\mathds{1}$. Here, we denoted with $\tau$ the optimizer of $\rho_1 \otimes \rho_2$. Moreover, we used that, for pure states, $\mathfrak{F}_0$ is equal to the stabilizer fidelity and the latter is multiplicative for an arbitrary number of pure states~\cite{bravyi2019simulation}. We also used that the partial trace of a free state is a free state since the partial trace is a free operation. In the last equality, we used equation~\eqref{one state}. Note that we always have $\mathfrak{F}_0(\rho_1 \otimes \rho_2) \geq \mathfrak{F}_0(\rho_1)\mathfrak{F}_0(\rho_2)$ since $\mathfrak{F}_0$ is super-multiplicative. Hence, we have that $\mathfrak{F}_0(\rho_1 \otimes \rho_2) = \mathfrak{F}_0(\rho_1)\mathfrak{F}_0(\rho_2)$. The proof for the tensor product of more than two qubits follows similarly. 
\end{proof}
We now show that all the  $\alpha$-$z$ R\'enyi relative entropies of magic are additive for single-qubit states in the range $|(1-\alpha)/z|=1$. The result for $\alpha=z=1/2$ implies that the stabilizer fidelity is multiplicative for tensor products of \text{any} single-qubit states.

\begin{theorem}
\label{Additivity in the line}
Let $(\alpha,z) \in \mathcal{D}$ such that $|(1-\alpha)|/z =1$ and $\{\rho_i\}_{i=1}^L$ be a set of single-qubit states. Then, we have
\begin{equation}
\mathfrak{D}_{\alpha,z}\left( \bigotimes_{i=1}^L \rho_i \right) = \sum_{i=1}^L \mathfrak{D}_{\alpha,z}(\rho_i) \,.
    \end{equation}
\end{theorem}
\begin{proof}
We first start by considering the tensor product of two-qubit states $\rho_1$ and $\rho_2$. To prove multiplicativity, we show that $\tau_1 \otimes \tau_2 \in \argmin_{\sigma \in \F} D_{\alpha,z}(\rho_1 \otimes \rho_2 \| \sigma)$ if $\tau_1$ and $\tau_2$ are some optimizers of the marginal problems, i.e., $\tau_1 \in \argmin_{\sigma \in \F} D_{\alpha,z}(\rho_1 \| \sigma)$ and $\tau_2 \in \argmin_{\sigma \in \F} D_{\alpha,z}(\rho_2 \| \sigma)$. Since the $D_{\alpha,z}$ are additive, the latter condition implies the additivity of the related monotones. Indeed, we have  $\mathfrak{D}_{\alpha,z}(\rho_1 \otimes \rho_2) = D_{\alpha,z}(\rho_1 \otimes \rho_2 \| \tau_1 \otimes \tau_2) = D_{\alpha,z}(\rho_1 \|\tau_1) + D_{\alpha,z}(\rho_2 \|\tau_2) = \mathfrak{D}_{\alpha,z}(\rho_1)+ \mathfrak{D}_{\alpha,z}(\rho_2)$. 
Therefore, to prove additivity, according to Theorem~\ref{main Theorem} in Appendix~\ref{necessary and sufficient}, we need to show that 
\begin{equation}
\label{functional}
\Trm(\sigma \Xi_{\alpha,z}(\rho_1\otimes \rho_2, \tau_1 \otimes \tau_2)) \leq Q_{\alpha,z}(\rho_1 \otimes \rho_2 \| \tau_1 \otimes \tau_2)= Q_{\alpha,z}(\rho_1 \| \tau_1) Q_{\alpha,z}(\rho_2 \| \tau_2), \quad \forall \sigma \in \F \,.
\end{equation}
We define the positive operator $\Xi_{\alpha,z}(\rho,\tau)$ in Appendix~\ref{necessary and sufficient}. In the last equality, we used that the $Q_{\alpha,z}$ are multiplicative. Note that we always have that $\tau_1 \otimes \tau_2$ satisfies the support conditions since, by assumption, the states $\tau_1$ and $\tau_2$ are optimizers of the marginal problems and hence satisfy them. 

We then have, for any $\sigma \in \F$,
\begin{align}
 \Trm(\sigma \Xi_{\alpha,z}(\rho_1 \otimes \rho_2, \tau_1 \otimes \tau_2)) & \leq  \max \limits_{\sigma \in \F}\Trm(\sigma \Xi_{\alpha,z}(\rho_1\otimes \rho_2, \tau_1 \otimes \tau_2))\\
 \label{tensor-product gradient}
 &  =  \max \limits_{\sigma \in \F}\Trm(\sigma\Xi_{\alpha,z}( \rho_1, \tau_1) \otimes \Xi_{\alpha,z}( \rho_2, \tau_2)) \\
&  =  C_{1,\alpha,z} C_{2,\alpha,z} \max \limits_{\sigma \in \F}\text{Tr}(\sigma\hat{\Xi}_{\alpha,z}( \rho_1, \tau_1) \otimes \hat{\Xi}_{\alpha,z}( \rho_2, \tau_2)) \\
 \label{additivity 1}
& =  C_{1,\alpha,z} C_{2,\alpha,z} \max \limits_{\sigma \in \F}\text{Tr}(\sigma\hat{\Xi}_{\alpha,z}( \rho_1, \tau_1)) \max \limits_{\sigma \in \F}\text{Tr}(\sigma\hat{\Xi}_{\alpha,z}( \rho_2, \tau_2)) \\
& =  \max \limits_{\sigma \in \F}\Trm(\sigma \Xi_{\alpha,z}( \rho_1, \tau_1)) \max \limits_{\sigma \in \F}\Trm(\sigma \Xi_{\alpha,z}( \rho_2, \tau_2)) \\
& = \; Q_{\alpha,z}(\rho_1\| \tau_1) Q_{\alpha,z}(\rho_2\| \tau_2) \,,
\end{align}
where we defined the positive constants $C_{i,\alpha,z}:= \text{Tr} (\Xi_{\alpha,z}( \rho_i, \tau_i) )$ for $i=1,2$, and for a positive semidefinite operator $A$ we defined the quantum state $\hat{A}:= A/\Trm(A)$. In~\eqref{tensor-product gradient} we used that for $|(1-\alpha)/z|=1$, it holds $\Xi_{\alpha,z}(\rho_1\otimes \rho_2, \tau_1 \otimes \tau_2) = \Xi_{\alpha,z}( \rho_1, \tau_1) \otimes \Xi_{\alpha,z}( \rho_2, \tau_2)$ (see Appendix~\ref{necessary and sufficient}).
In~\eqref{additivity 1} we used the multiplicativity result in Proposition~\ref{multiplicativity F_0}. The last equality follows from the fact that $\tau_1$ and $\tau_2$ are, by assumptions, optimizers of the marginal problems. Indeed, from Theorem~\ref{main Theorem}, if $\tau \in \argmin_{\sigma \in \F} Q_{\alpha,z}(\rho \| \sigma)$, we have that $\Trm(\sigma\Xi_{\alpha,z}( \rho, \tau)) \leq Q_{\alpha,z}(\rho\| \tau)$ for any $\sigma \in \F$.  It is easy to see that the latter inequality is saturated for $\sigma = \tau$.
The latter chain of inequalities proves the inequality~\eqref{functional}. The proof for more than two qubits follows similarly. 
\end{proof}

In the next section, we prove that for states that commute with their optimizer, we can prove additivity for all the range $(\alpha,z) \in \mathcal{D}$. In particular, the latter range contains the relative entropy of magic. We show in the Appendix that the single-qubit states that commute with its optimizer are the noisy (depolarized) $T$,$H$, and $F$ states. 

\section{Additivity of all the $\alpha$-$z$ R\'enyi relative entropies of magic for a specific class of single-qubit states}
\label{relative entropy additivity}
In this section, we show that the relative entropy of magic is additive for single-qubit states when all of them commute with their optimizer. Note that we actually show something slightly stronger. Indeed, we do not require one of the states to commute with its optimizer. In Section~\ref{T, H and F states}, we show that the single-qubit states that commute with its optimizer are only the $T,H$, or $F$ states subject to depolarizing noise. Previously, additivity was known only for pure $H,T$, and $F$ states. Indeed, in this case, the relative entropy of magic equals both the stabilizer fidelity and the generalized robustness of magic, which are both additive for single-qubit pure states~\cite{bravyi2019simulation,seddon2021quantifying}.  

We extend this result to include all the $\alpha$-$z$ R\'enyi relative entropy of magic for all $(\alpha,z) \in \mathcal{D}$. At the end of the section, we show that there exist two single-qubit states for which the relative entropy of magic is not additive. This shows that, unlike the fidelity or the robustness measure, the relative entropy one is not additive for tensor products of any single-qubit states. 
\begin{theorem}
\label{theorem relative entropy of magic qubits}
Let $(\alpha,z) \in \mathcal{D}$, $\{\rho_i\}_{i=L}^L$ be a set of single-qubit states and  $\tau_i \in \argmin_{\sigma \in \F} D_{\alpha,z}(\rho_i \| \sigma)$. If $[\rho_i,\tau_i]=0$ for any $i \in \{1,...,L-1\}$, then we have that  
\begin{equation}
\mathfrak{D}_{\alpha,z}\left( \bigotimes_{i=1}^L \rho_i \right) = \sum_{i=1}^L \mathfrak{D}_{\alpha,z}(\rho_i) \,.
    \end{equation}
\end{theorem}
\begin{proof}
The proof is similar to the one of Theorem~\ref{Additivity in the line}. We first start by considering the tensor product of two single-qubit states $\rho_1$ and $\rho_2$. In particular, we show that if $\tau_1$ is an optimizer of $\rho_1$ such that $[\rho_1,\tau_1]=0$, then an optimizer of $\rho_1 \otimes \rho_2$ is $\tau_1 \otimes \tau_2$ where $\tau_2$ is any optimizer of $\rho_2$. 
According to Theorem~\ref{main Theorem}, we want to show that 
\begin{equation}
\label{functional RM}
\Trm(\sigma \Xi_{\alpha,z}(\rho_1\otimes \rho_2, \tau_1 \otimes \tau_2)) \leq  Q_{\alpha}(\rho_1\| \tau_1) Q_{\alpha,z}(\rho_2\| \tau_2)\,, \quad \forall \sigma \in \F\,,
\end{equation}
where we denoted with $\tau_1$ and $\tau_2$ some optimizers of $\rho_1$ and $\rho_2$, respectively.  

We now show that if $[\rho_1,\tau_1]=0$, then we have $ \Xi_{\alpha,z}(\rho_1\otimes \rho_2, \tau_1 \otimes \tau_2) = \Xi_{\alpha,z}(\rho_1, \tau_1) \otimes \Xi_{\alpha,z}(\rho_2, \tau_2)$. This fact has been pointed out in the proof of Theorem 7 in~\cite{rubboli2024new}. We reprove it here for completeness. 
We use the spectral decomposition $\tau_1 = \sum_l t_{1,l}|\psi_l\rangle \! \langle \psi_l|$ and $\tau_2 = \sum_r t_{2,r}|\xi_r\rangle \! \langle \xi_r|$. Here, $|\psi_l\rangle$ is the common eigenbasis of $\rho_1$ and $\tau_1$.   We have
\begin{align}
 & \Xi_{\alpha,z}(\rho_1 \otimes \rho_2, \tau_1 \otimes \tau_2) \\
&\qquad = \, K_{\alpha,z} \sum \limits_{l, r_1,r_2}  \int_0^\infty \left(t_{1,l} t_{2,r_1} +k\right)^{-1} \left(t_{1,l} t_{2,r_2} +k\right)^{-1} k^\frac{1-\alpha}{z}\d k \;   \;| \psi_{l} \rangle \! \langle \psi_{l}| \chi_{\alpha,z}(\rho_1,\tau_1) |\psi_{l} \rangle \! \langle \psi_{l}| \;  \notag\\
& \qquad \quad  \otimes | \xi_{r_1} \rangle \! \langle \xi_{r_1}| \chi_{\alpha,z}(\rho_2,\tau_2) |\xi_{r_2} \rangle \! \langle \xi_{r_2}|  \\
& \qquad = \, K_{\alpha,z} \sum \limits_{l} t_{1,l}^{\frac{1-\alpha}{z}-1} | \psi_{l} \rangle \! \langle \psi_{l}| \chi_{\alpha,z}(\rho_1,\tau_1) |\psi_{l} \rangle \! \langle \psi_{l}| \;  \notag \\
& \qquad \quad \otimes  \sum_{r_1,r_2} \int_0^\infty \left( t_{2,r_1} +k\right)^{-1} \left(t_{2,r_2} +k\right)^{-1} k^\frac{1-\alpha}{z}\d k |\xi_{r_1} \rangle \! \langle \xi_{r_1}| \chi_{\alpha,z}(\rho_2,\tau_2) |\xi_{r_2} \rangle \! \langle \xi_{r_2}| \\
\label{penultima}
&\qquad = \; \rho_1^\alpha\tau_1^{-\alpha} \otimes \Xi_{\alpha,z}( \rho_2, \tau_2)\\
& \qquad  =\;\Xi_{\alpha,z}( \rho_1, \tau_1) \otimes \Xi_{\alpha,z}( \rho_2, \tau_2) \,,
\end{align}
where in the second equality we changed the measure $k \rightarrow t_{1,l} k$. Moreover, in~\eqref{penultima} we used that $ \sum_{k} t_{1,l}^{\frac{1-\alpha}{z}-1} | \psi_{l} \rangle \! \langle \psi_{l}| \chi_{\alpha,z}(\rho_1,\tau_1) |\psi_{l} \rangle \! \langle \psi_{l}| = \rho_1^\alpha\tau_1^{-\alpha}$. In the last equality, we used that if $[\rho_1,\tau_1]=0$ then $\Xi_{\alpha,z}( \rho_1, \tau_1)=\rho^\alpha_1 \tau_1^{-\alpha}$. 

Following the same arguments to the proof of Theorem~\ref{Additivity in the line}, the above result implies that  for any $\sigma \in \F$ we have 
$
 \Tr(\sigma \Xi_{\alpha,z}(\rho_1 \otimes \rho_2, \tau_1 \otimes \tau_2))  \leq  Q_{\alpha}(\rho_1\| \tau_1) Q_{\alpha,z}(\rho_2\| \tau_2)$.
The latter inequality proves the inequality~\eqref{functional RM}. The proof for more than two qubits follows similarly by noting that $\Xi_{\alpha,z}(\bigotimes_{i=1}^L\rho_i, \bigotimes_{i=1}^L\tau_i) = \bigotimes_{i=1}^L \Xi_{\alpha,z}(\rho_i, \tau_i)$ if $[\rho_i,\tau_i]=0$ for any $i \in \{1,\hdots,L-1\}$ . 
\end{proof}

\subsection{Counterexample to the additivity of the relative entropy of magic for two multi-qubit states}
\label{counterexample qubits}
We give new examples of the non-additivity of the relative entropy of magic for single-qubit states. The relative entropy of magic can be computed via SDP approximations following the methods discussed in~\cite{fawzi2019semidefinite}. Let us consider the pure state $\rho$ with Bloch vector $
\vec{r} = ((6+2\sqrt{3})^{-\frac{1}{2}}, (6+2\sqrt{3})^{-\frac{1}{2}},(3-\sqrt{3})^{-\frac{1}{2}})
$.
We obtain numerically $\mathfrak{D}(\rho \otimes \rho) - 2\mathfrak{D}(\rho) \approx -0.002121$. This shows that the relative entropy of magic is not additive for general two single-qubit states. However, we proved in Theorem~\ref{theorem relative entropy of magic qubits} that additivity holds if at least one of the two states belongs to a symmetry axis of the stabilizer octahedron. We also note that at most one of the states can be of general form. Indeed, we numerically find $\mathfrak{D}(\ketbra{F}{F} \otimes \rho \otimes \rho) - (\mathfrak{D}(\ketbra{F}{F})+2\mathfrak{D}(\rho)) \approx -0.002121$. Here, we denoted with $\ketbra{F}{F}$ the pure state defined in equation~\eqref{T, H and F states} and we take $\rho$ as above. 
Further examples that violate the additivity property could be easily found also for multi-qubit states. 

\section{Additivity of all the $\alpha$-$z$ R\'enyi relative entropies of magic for a specific class of two and three-qubit states}
\label{multi-qubit additivity}
In this section, we prove that for  a specific class of mixed two and three-qubit states, all the $\alpha$-$z$ R\'enyi relative entropies of magic are additive. Previously, it was known that all the monotones in the line $|(1-\alpha)/z|=1$ are additive for pure states that describe a system of at most three qubits. Indeed, in the case of pure states, for $z=1-\alpha$ and $z=\alpha-1$, the $\alpha$-$z$ R\'enyi relative entropies of magic equal the log-stabilizer fidelity and the generalized log-robustness of magic, respectively. Then additivity is a consequence of multiplicativity of the stabilizer fidelity~\cite[Theorem 5]{bravyi2019simulation} and the one of the stabilizer extent~\cite[Proposition 1]{bravyi2019simulation} which for pure states equals the generalized robustness~\cite{regula2017convex}. Since any monotone based on a quantum relative entropy is contained between $\mathfrak{D}_{\min}$ and $\mathfrak{D}_{\max}$~\cite{gour2020optimal}, for pure states $\rho$ such that $\mathfrak{D}_{\min}(\rho) = \mathfrak{D}_{\max}(\rho)$, the latter result can be extended to any monotone based on a quantum relative entropy.

In the following, we address a specific class of mixed two and three-qubit states. In Section~\ref{specific states}, we show that the latter class includes the Toffoli, Hoggar, and $CS$ states subject to global depolarizing noise. We do not know whether there exist other two and three-qubit states that belong to this class.

\begin{theorem}
\label{theorem relative entropy of magic}
Let $(\alpha,z) \in \mathcal{D}$ and $\{\rho_i\}_{i=L}^L$ be a set of states that describe a system of at most three qubits such that $\rho_i = \Delta_{p_i}(\ketbra{\psi_i}{\psi_i})$ for some $p_i\geq 0$ and a pure state $\ketbra{\psi_i}{\psi_i}$ for all $i=1,..,L$. Moreover, let  $\tau_i \in \argmin_{\sigma \in \F} D_{\alpha,z}(\rho_i \| \sigma)$. If $\tau_i = \Delta_{t_i}(\rho_i)$ for some $t_i\geq0$ for all $i=1,..,L$, then we have that  
\begin{equation}
\mathfrak{D}_{\alpha,z}\left( \bigotimes_{i=1}^L \rho_i \right) = \sum_{i=1}^L \mathfrak{D}_{\alpha,z}(\rho_i) \,.
    \end{equation}
\end{theorem}
\begin{proof}
By assumption, we have $\rho_i = \Delta_{p_i}(\ketbra{\psi_i}{\psi_i})$ and  $\tau_i = \Delta_{t_i}(\ketbra{\psi_i}{\psi_i})$ for some state $\ketbra{\psi_i}{\psi_i}$ and parameters $p_i$ and $t_i$. Since $[\rho_i,\tau_i]=0$, we check the conditions of Corollary~\ref{commuting}. We have
\begin{align}
\Xi_{\alpha}(\rho_i,\tau_i) = (a_i-b_i)\ketbra{\psi_i}{\psi_i}+b_i \mathds{1} \,,
\end{align}
where 
\begin{equation}
a_i = \left(\frac{1+(d-1)p_i}{1+(d-1)t_i}\right)^\alpha, \quad 
b_i = \left(\frac{1-p_i}{1-t_i}\right)^\alpha \,.
\end{equation}
Note that, as discussed at the end of Section~\ref{Magic monotones} we can restrict ourselves to the case where $\rho_i$ are magic. In this case, we have that $p_i>t_i$ and $a_i>b_i$. Moreover, since the states commute, we have that $\Xi_\alpha(\bigotimes_{i=1}^L\rho_i, \bigotimes_{i=1}^L\tau_i) = \bigotimes_{i=1}^L \Xi_\alpha(\rho_i, \tau_i)$. The proof then follows the same steps as the one given for Theorem~\ref{Additivity in the line}. 
\end{proof}
Note that, in contrast to the single-qubit case (see Theorem~\ref{theorem relative entropy of magic qubits}), we require that \textit{all} states commute with their optimizer. We also remark that the above theorem can be strengthened for the range $|(1-\alpha)/z|=1$. Indeed, in this range, for additivity to hold, we do not require any constraints for the single-qubit states (but the two and three-qubit states must satisfy the same constraints of Theorem~\ref{theorem relative entropy of magic}).

We also mention that in~\cite[Claim 2]{bravyi2019simulation}, the authors proved that the stabilizer fidelity is generally not multiplicative for pure states for sufficiently large dimensions. The latter results have been extended to the generalized robustness in~\cite{heimendahl2021stabilizer}. We numerically find that the bounds given in these references become non-trivial for twelve-qubit states or larger.   

It is still an open problem whether the $\alpha$-$z$ R\'enyi relative entropies of magic in the range $|(1-\alpha)/z|=1$ are additive for all two and three-qubit states. Numerical simulations suggest that additivity holds. We discuss in Section~\ref{Further directions} a possible direction to prove this result. 

\section{Bounds for magic state distillation}
\label{Magic state transformation}
The $\alpha$-$z$ R\'enyi relative entropies are monotone under free probabilistic protocols (see Appendix~\ref{distillation lower bound} for more details). This allows us to derive lower bounds for the number of copies of the input state needed for free probabilistic transformations. Interestingly, we obtain tighter bounds than the existing ones in the literature. Let us consider the free transformation that maps $\rho^{\otimes n}$ into $ \tau$ and $\eta$ with probability $p$ and $1-p$, respectively. Then, Corollary~\ref{strong monotonicity} and subadditivity imply that
\begin{align}
\label{distillation monotonocity}
n\mathfrak{D}_{\alpha,z}(\rho) \geq \frac{\alpha}{\alpha-1}\log{\left(p\mathcal{Q}^\frac{1}{\alpha}_{\alpha,z}(\tau)+(1-p)\right)} \,,
\end{align}
where we also used that $\mathcal{Q}_{\alpha,z}(\eta) \leq 1$ for $\alpha \leq 1$ and the opposite inequality holds for $\alpha \geq 1$.
The latter relation provides a lower bound for the required number of copies $n$.

The results we derived in the previous sections allow us to compute efficiently $\mathcal{Q}_{\alpha,z}(\tau)$ for the case where $\tau$ consists of tensor products of states for which additivity holds. Moreover, the same results allow us to consider also situations where a catalyst is involved in the process. Indeed, catalysts could potentially activate some transformations as shown in~\cite{campbell2011catalysis}, and improve the distillation overhead. We recall that a catalyst is an ancillary system that is returned unchanged at the end of the process. Explicitly, it is straightforward to verify that the bounds apply for processes where a catalyst $\nu$ is returned with probability one and satisfies the additivity  property $\mathfrak{D}_{\alpha,z}(\tau \otimes \nu) = \mathfrak{D}_{\alpha,z}(\tau)+\mathfrak{D}_{\alpha,z}(\nu)$. 

In a distillation protocol, one aims to transform copies of a less resourceful state into copies of a more resourceful state. In the following, let us consider the distillation protocol $\Delta_{3/4}(\ketbra{T}{T})^{\otimes k} \rightarrow \ketbra{T}{T}$ with probability $p$ and fidelity $1-\varepsilon$. From standard symmetry argument (see e.g.,~\cite{horodecki1999reduction} and~~\cite{regula2022probabilistic}), it is sufficient to restrict the analysis to the exact transformation $\Delta_{3/4}(\ketbra{T}{T})^{\otimes n} \rightarrow \Delta_{1-2\varepsilon}(\ketbra{T}{T})$ with probability $p$.  

In this specific case, we numerically found that the limit $\alpha \rightarrow 0, z=1-\alpha$ in~\eqref{distillation monotonocity} is the one that gives the tightest lower bounds. 
We refer to this bound as $\mathfrak{D}_{\min}$ bound.
We also noticed that the monotones $\mathfrak{D}_{\alpha,z}^{\rm{rev}}(\rho) := \min_{\sigma \in \F} D_{\alpha,z}(\sigma \| \rho)$ yield stronger bounds in the deterministic case ($p=1$). In particular, for the case we consider, the best bound is provided by a value of $\alpha$ that approaches one from above in the limit $\varepsilon \rightarrow 0$. However, we numerically observe that these monotones are not in general additive for $\alpha \geq 1$ and we leave a more detailed analysis for further work.

In Fig.~\ref{fig: distillation} we plot the $\mathfrak{D}_{\min}$ bound and the $\mathfrak{D}^{\rm{rev}}_{\alpha,z}$ bound and compare them to the weight one given in~\cite{regula2021fundamental} and the projective robustness one given in~\cite{regula2022probabilistic}. To the best of our knowledge, the weight bound is the tightest known deterministic bound. The projective robustness bound is essentially the best-known probabilistic one and holds for any $p>0$. In the deterministic case, the $\mathfrak{D}^{\rm{rev}}_{\alpha,z}$ bound performs the best and both the $\mathfrak{D}_{\min}$ and the $\mathfrak{D}^{\rm{rev}}_{\alpha,z}$ bounds perform orders of magnitude better than both the weight and the projective robustness ones. This shows that the projective robustness bound is tight only for small probabilities. The code is available on GitHub.~\footnote{\href{https://github.com/RobertoRubboli/Additivity-magic-monotones}{https://github.com/RobertoRubboli/Additivity-magic-monotones}}

\begin{figure}
\includegraphics[scale=.9]{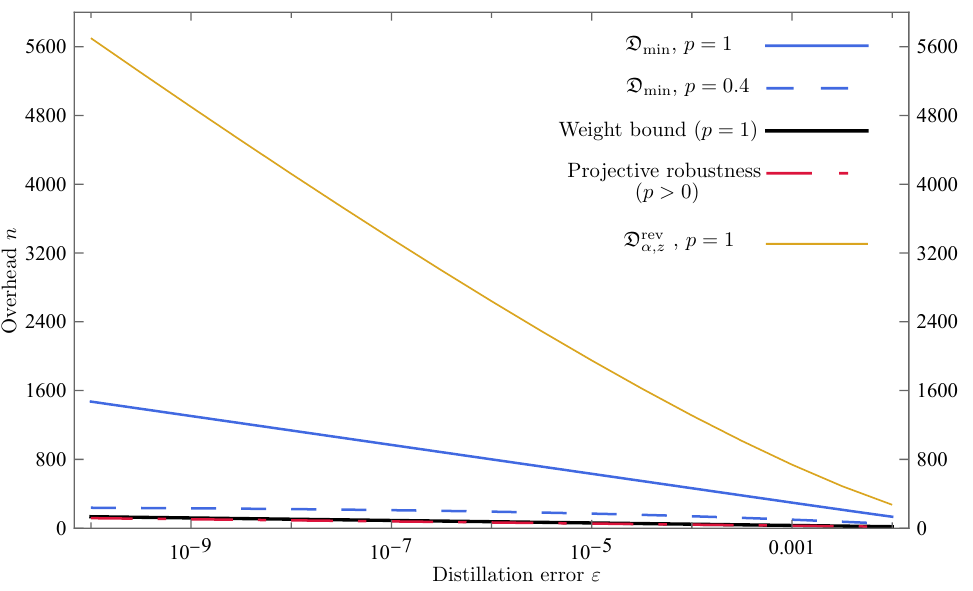}
\caption{The graph compares multiple lower bounds for the overhead of $T$ state distillation. In particular, we consider the transformation $\Delta_{3/4}{(\ketbra{T}{T})} \rightarrow \ketbra{T}{T}$ with output fidelity at least $1-\varepsilon$. We plot the number $n$ of copies needed as a function of the final target error $\varepsilon$. We show in blue the $\mathfrak{D}_{\min}$ bound for the values of the probability $p =1,0.4$. In yellow, we show the $\mathfrak{D}^{\rm{rev}}_{\alpha,z}$ bound for $p=1$. In black, we show the weight bound~\cite{regula2021fundamental} which is the best-known deterministic bound. In red, we show the projective robustness bound~\cite{regula2022probabilistic} which holds for any $p>0$. In the deterministic case, our $\mathfrak{D}_{\min}$ and $\mathfrak{D}^{\rm{rev}}_{\alpha,z}$ bounds are stronger than both the previously known weight and the projective robustness bounds.}
\label{fig: distillation}
\end{figure}

\subsection{Asymptotic transformations}
The asymptotic transformation rate $R(\rho\to\tau)$ from $\rho$ to $\tau$ is defined as 
\begin{equation}
    R(\rho\to\tau):= \sup\left\{r:\exists \{\mathcal{E}_n\}_n\mbox{ s.t. }\lim_{n\to\infty}\|\mathcal{E}_n(\rho^{\otimes n})-\tau^{\otimes rn}\|_1=0,\  \mathcal{E}_n\in\mathcal{O}_\mathcal{F}\right\}.
\end{equation}
A general upper bound for the asymptotic transformation rate is given by~\cite{horodecki2013quantumness}
$
R(\rho\to\tau)\leq \mathfrak{D}^\infty(\rho)/\mathfrak{D}^\infty(\tau)
$, where $\mathfrak{D}^\infty(\rho):=\lim_{n\to\infty}\mathfrak{D}(\rho^{\otimes n})/n$ is the regularized relative entropy of magic. However, the latter quantity is generally not computable due to the regularization. Observing that $ -\log\mathfrak{F}(\tau)\leq \mathfrak{D}(\tau)$ and $\mathfrak{D}^\infty(\tau)\leq \mathfrak{D}(\tau)$ for an arbitrary state $\tau$, the additivity of $-\log\mathfrak{F}$ for an arbitrary single-qubit state shown in Theorem~\ref{Additivity in the line} now gives 
\begin{equation}
    R(\rho\to\tau) \leq \frac{\mathfrak{D}(\rho)}{-\log\mathfrak{F}(\tau)}
\end{equation}
for an arbitrary state $\rho$ and any single-qubit state $\tau$. This extends the result in~\cite[Section VI]{seddon2021quantifying}, which is restricted to a pure target state $\tau$.

Moreover, if $\tau$ is a state that describes a system of at most three qubits and satisfies the conditions of Theorem~\ref{theorem relative entropy of magic}, we get the stronger upper bound 
\begin{equation}
 R(\rho\to\tau) \leq \frac{\mathfrak{D}(\rho)}{\mathfrak{D}(\tau)}\,.
\end{equation}

\section{Closed-form expression of the stabilizer fidelity and generalized robustness of magic for single-qubit states}
\label{closed forms}
In this section, we show that two additive quantities for single-qubit states, namely the generalized robustness and the stabilizer fidelity, have a closed form for all single-qubit states. To the best of our knowledge, this is the first closed-form solution result of a magic monotone based on a quantum relative entropy that includes all single-qubit states. 

Due to the symmetry of the stabilizer polytope in the Bloch sphere, without loss of generality, we restrict ourselves to states in the positive octant of the Bloch sphere.   We split the positive octant into three regions 
\begin{align}
\label{regions}
P_x:=\{\rho : r_x \leq r_y, r_x \leq r_z\} \,, \quad
P_y:=\{\rho : r_y \leq r_x, r_y \leq r_z\} \,, \quad
P_z:= \{\rho : r_z \leq r_x, r_z \leq r_y\} \,.
\end{align}  
We denote the norm and the $l_1$-norm of a vector as $|\vec{r}|=(\sum_ir_i^2)^{\frac{1}{2}}$ and $\|\vec{r}\|_1:=\sum_{i}|r_i|$, respectively.
\subsection{Generalized robustness of magic}
We start by deriving a closed-form expression of the generalized robustness of magic for single-qubit states. 
\begin{proposition}
\label{Robustnes}
Let $\rho \notin \F $ be a single-qubit state whose Bloch vector $\vec{r}$ lies in the positive octant. If $\rho \in P_x $ then
\begin{equation}
\Lambda^+(\rho) = \begin{cases}
\frac{1}{\sqrt{3}+1}(\sqrt{3}+\|\vec{r}\|_1) & \text{if}\quad g(\vec{r}) \geq 0  \\
2+r_x-\|\vec{r}\|_1+\sqrt{2\left(r_x^2+(1+r_x-\|\vec{r}\|_1)^2\right )} & \text{if}\quad g(\vec{r}) \leq 0
\end{cases} \,,
\end{equation}
where $g(\vec{r}) : = (3+\sqrt{3})r_x+1-\|\vec{r}\|_1$. For the regions $P_y$ and $P_z$ the results can be obtained by exchanging $r_x \leftrightarrow r_y$ and $r_x \leftrightarrow r_z$, respectively. 
\end{proposition}
We show in Fig.~\ref{figrob} the plane $g(\vec{r})=0$.

\begin{figure}
     \centering
     \begin{subfigure}[h]{0.48\textwidth}
         \includegraphics[width=\textwidth]{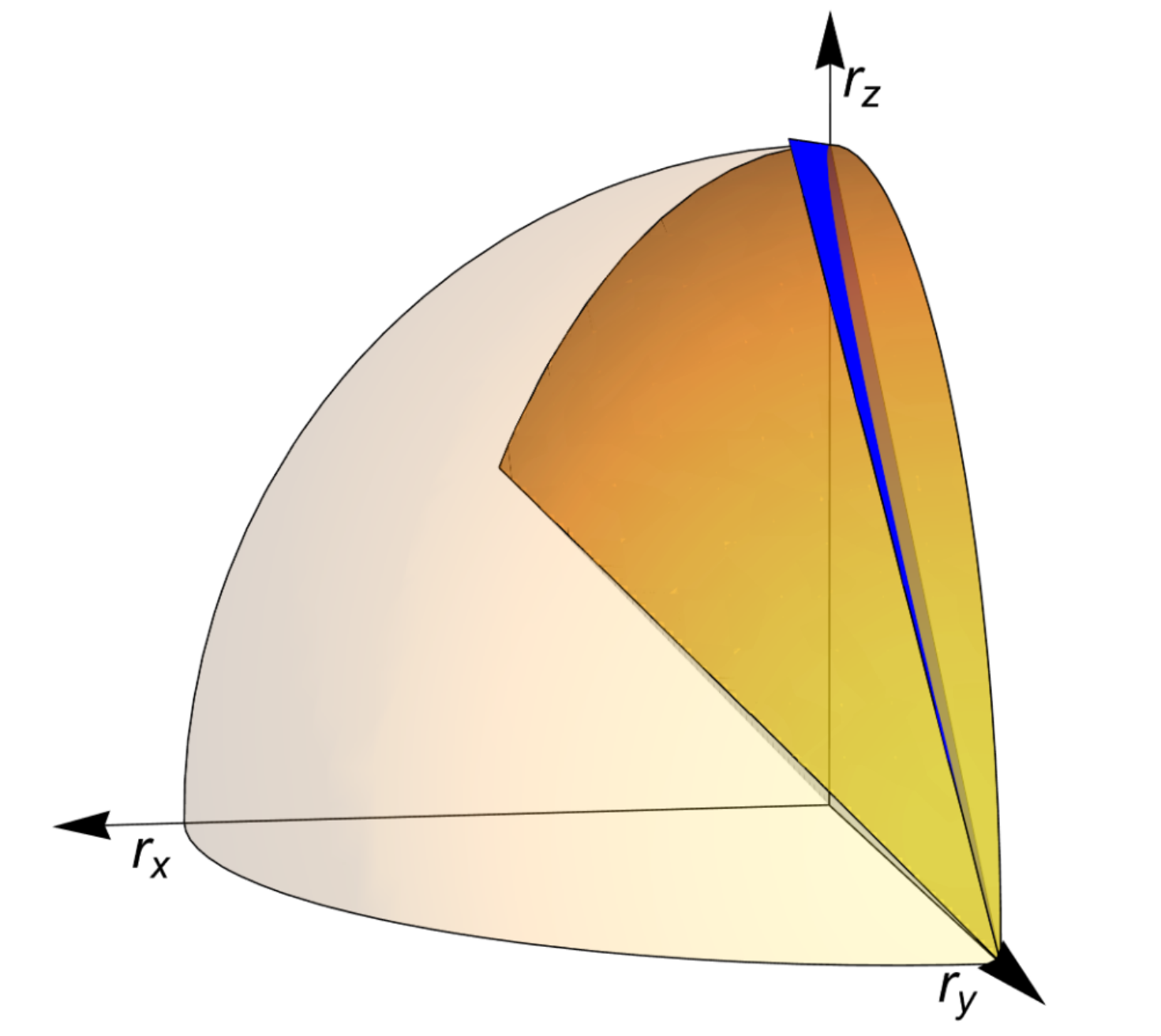}
         \caption{}
     \end{subfigure}
     \begin{subfigure}[h]{0.5\textwidth}
         \centering
         \includegraphics[width=\textwidth]{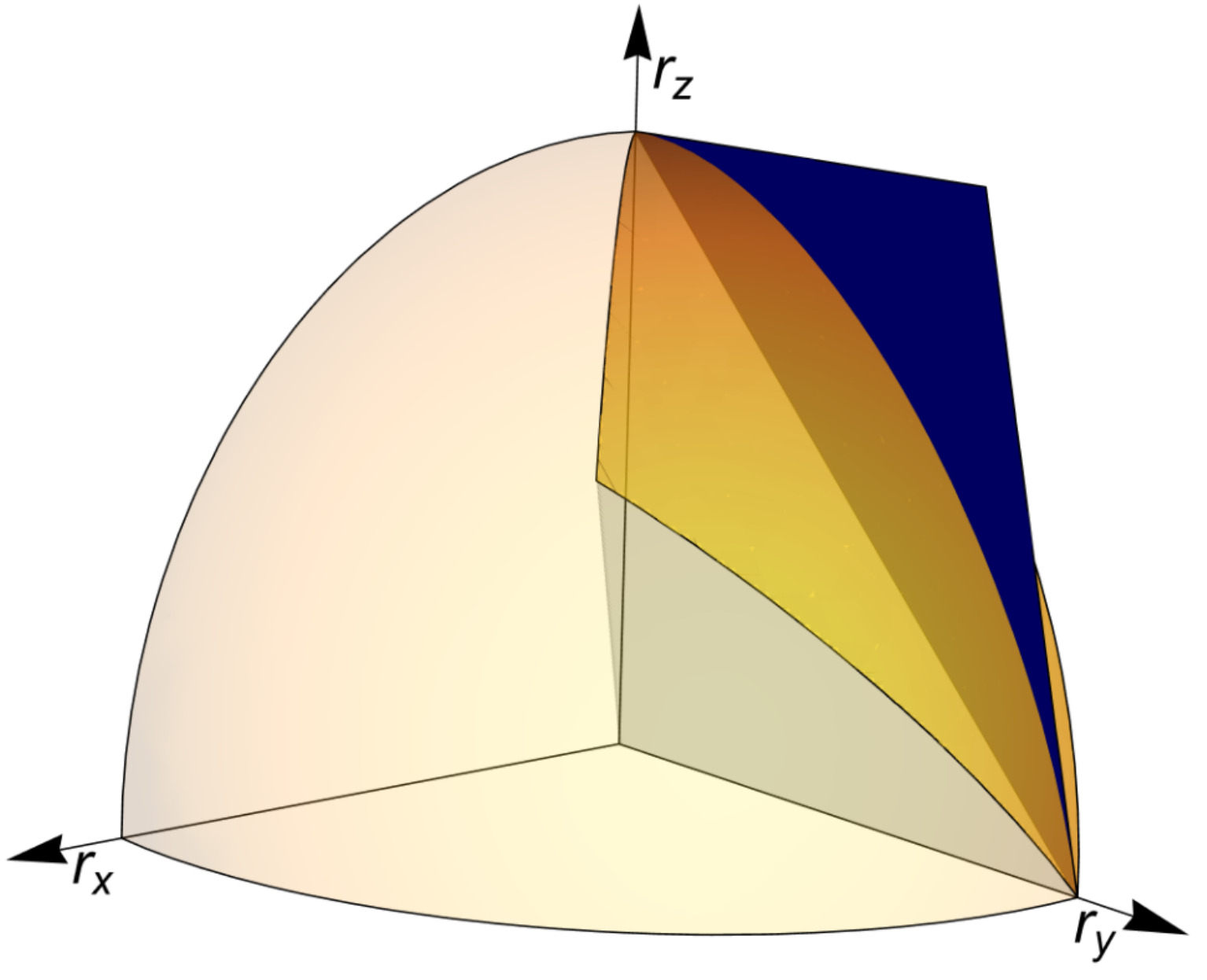}
         \caption{}
     \end{subfigure}
     \\
     \begin{subfigure}[h]{0.52\textwidth}
         \includegraphics[width=\textwidth]{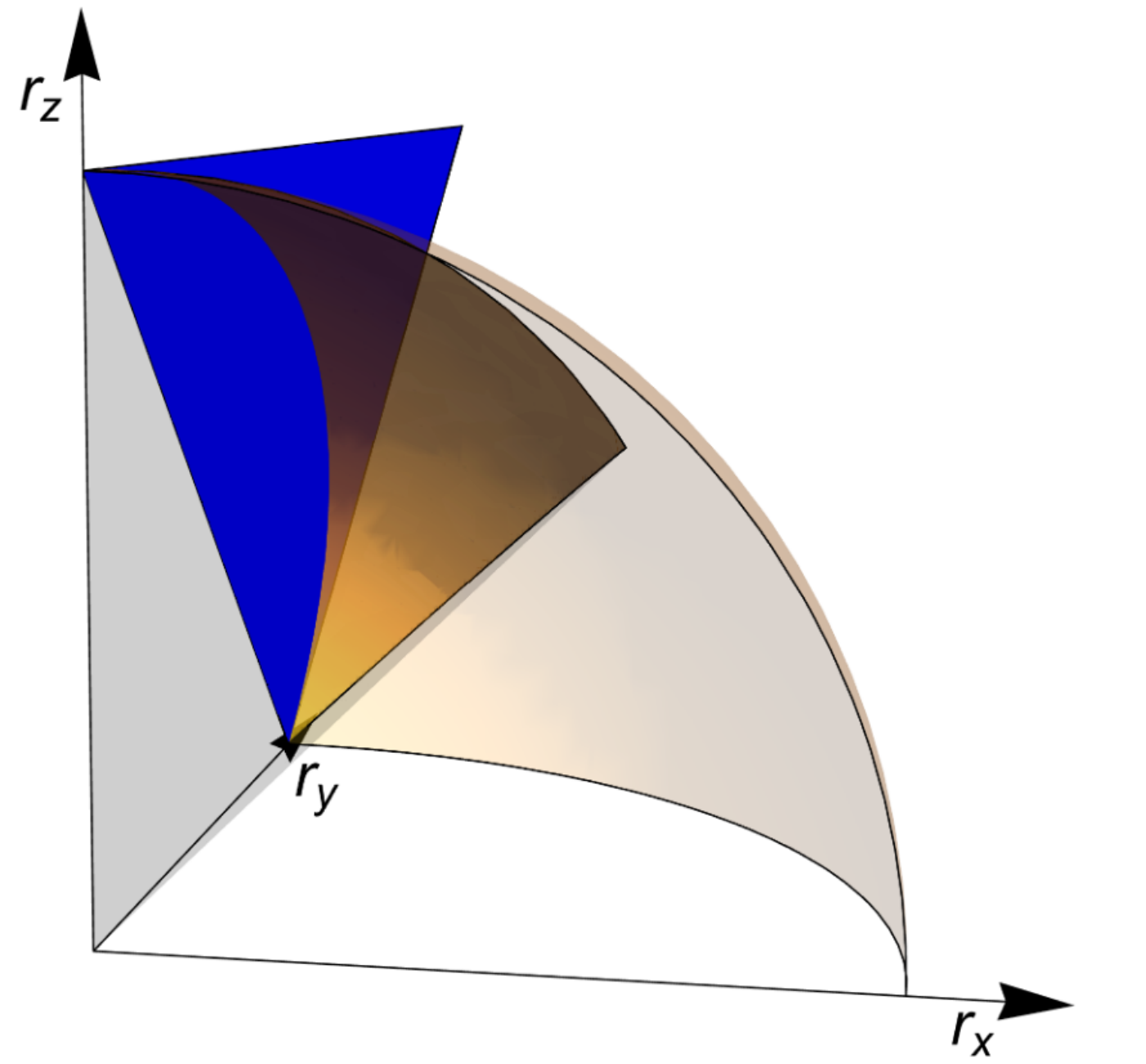}
         \caption{}
     \end{subfigure}
     \begin{subfigure}[h]{0.45\textwidth}
         \centering
         \includegraphics[width=\textwidth]{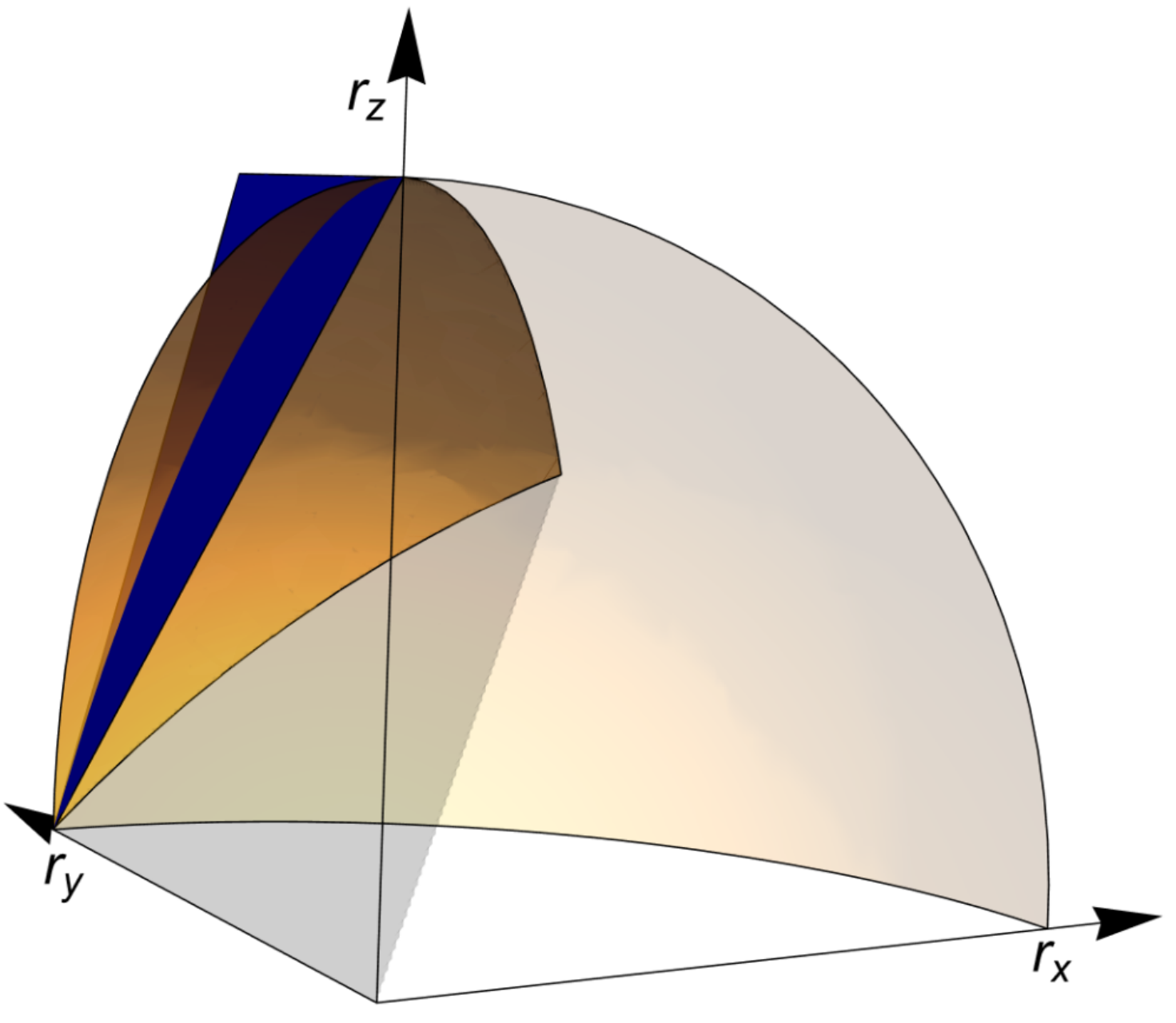}
         \caption{}
     \end{subfigure}
\caption{The figure represents four different angles of the plane $g(\vec{r})=0$ defined in Proposition~\ref{Robustnes} for the positive octant of the Bloch sphere. The shaded black area is the region $P_x$ defined in equation~\eqref{regions}. In the region between the plane $g(\vec{r})=0$ and the $z$-$y$ plane of the Bloch sphere it holds $g(\vec{r})\leq 0$. The opposite condition holds on the other side of the plane.}
\label{figrob}
\end{figure}

\begin{proof}
We consider the surface with constant $D_{\max}(\rho \| \sigma)$. Here, we think of $\rho$ as a fixed input state. We set $k=\log{D_{\max}(\rho\|\sigma)}$ and we denote the Bloch vectors of $\rho$ and $\sigma$ as $\vec{r}$ and $\vec{s}$, respectively. We also assume that $\rho \in P_x$. We define $a(\vec{s}) = 1-|\vec{s}|^2$ and $b(\vec{s}) = 1-\vec{r}\cdot\vec{s}$ and $d=1-|\vec{r}|^2$. Using Lemma~\ref{robustness qubit} we have 
\begin{equation}
ka(\vec{s})-b(\vec{s}) = \sqrt{b(\vec{s})^2-da(\vec{s})} \,.
\end{equation}
We then take the square and rearrange the terms to get
\begin{equation}
k^2a(\vec{s})-2kb(\vec{s})+d=0 \,.
\end{equation}
Using the explicit form of the coefficients $a(\vec{s})$ and $b(\vec{s})$ we obtain 
\begin{equation}
\left(s_x-\frac{r_x}{k}\right)^2+\left(s_y-\frac{r_y}{k}\right)^2+\left(s_z-\frac{r_z}{k}\right)^2 = \left(1-\frac{1}{k}\right)^2 \,.
\end{equation}
The above equation is a sphere with a center that goes to zero and with a radius that goes to $1$ as $k$ grows to infinity. We denote the coordinates of the centers $r_{ki}=r_i/k$. To find the smallest $k$ at which the spheres intersect the stabilizer polytope, we first calculate the projection of the points $(r_{kx},r_{ky},r_{kz})$ onto the stabilizer plane $s_x+s_y+s_z=1$. We then compute the distance between the projection and the point and set it equal to the value of the radius. Finally, we then solve for $k$. To find the projection we set
\begin{equation}
(x_k,y_k,z_k) = (r_{kx},r_{ky},r_{kz}) +\lambda (1,1,1) \,,
\end{equation}
which gives by requiring $x_k+y_k+z_k=1$ the value $\lambda = (1-r_{kx}-r_{ky}-r_{kz})/3$. Hence, the projection is
\begin{align}
&(x_k,y_k,z_k) \\
& \;= \left(\frac{2}{3}r_{kx}+\frac{1}{3}(1-r_{ky}-r_{kz}), \frac{2}{3}r_{ky}+\frac{1}{3}(1-r_{kx}-r_{kz}), \frac{2}{3}r_{kz}+\frac{1}{3}(1-r_{ky}-r_{kx} )\right) \,.
\end{align} 
By setting the distance between the center and the projection equal to the radius, we obtain the equation
\begin{equation}
1-\frac{1}{k^*}=-\frac{1}{\sqrt{3}}(1-r_{k^*x}-r_{k^*y}-r_{k^*z})\,.
\end{equation}
The reason for the minus sign when we calculate the square root on the r.h.s. is that, as it is easy to check with the value of $k^*$ below, we always have $r_{ky}+r_{kx}+r_{kz} \geq 1$ for any $k \leq k^*$ since $r_x+r_y+r_z\geq 1$ for magic states.
The above equation gives for $k^*=2^{\mathfrak{D}_{\max}(\rho)}$ the value
\begin{equation}
k^*= \frac{1}{\sqrt{3}+1}(\sqrt{3}+r_x+r_y+r_z) \,.
\end{equation}
We now distinguish two cases. The first case is the one for which the projection is inside the Bloch sphere, i.e., $x_k^* \geq 0$ (note that inside $P_x$ we have that $x_k$ is always smaller than $y_k$ and $z_k$). In this case the result above holds since the projection is inside the Bloch sphere. Using the value of $k^*$, the condition $x_{k^*}\geq 0$ gives
\begin{equation}
r_x(2+\sqrt{3})+1-r_y-r_z \geq 0 \,.
\end{equation} 

The second case is the one for which using the above value of $k^*$, the projection onto the plane lies outside the Bloch sphere, i.e., $x_k^*\leq 0$. In this case, since the optimizer is on the $y$-$z$ plane, we calculate the projection with the stabilizer straight line $s_z=1-s_y$, $sx=0$. We first construct the plane orthonormal to the line $s_z=1-s_y$ and find the constant parameter $d$ for which it passes through the point $(r_{kx},r_{ky},r_{kz})$. The equation of the plane is $s_y-s_z+d=0$ which gives $d=r_{kz}-r_{ky}$. We set $s_y=t$ and $s_z=1-t$ and substitute in the equation of the plane to find $t=(r_{ky}-r_{kz}+1)/2$ which gives the projection $y_k=(r_{ky}-r_{kz}+1)/2$ and $z_k=(r_{kz}-r_{ky}+1)/2$ (note that they are positive since each Bloch component is smaller than $1$ and $k \geq 1$). Setting again the distance between the point and the projection equal to the radius, solving the second-order equation, and discarding the negative solution, we find
\begin{equation}
k^* = 2-r_y-r_z+\sqrt{2r_x^2+2(1-r_y-r_z)^2} \,.
\end{equation}

\end{proof}

\subsection{Stabilizer fidelity}
We now give a closed-form expression of the stabilizer fidelity for single-qubit states. 
\begin{proposition}
\label{Fidelity}
Let $\rho \notin \F$ be a single-qubit state with Bloch vector $\vec{r}$ in the positive octant. If $\rho \in P_x $ then
\begin{equation}
\mathfrak{F}(\rho) = \begin{cases}
\frac{1}{6}\left(3+\|\vec{r}\|_1+\sqrt{6-2\|\vec{r}\|_1^2}\right) & \text{if}\quad f(\vec{r}) \geq 0  \\
\frac{1}{4}\left(2+\|\vec{r}\|_1-r_x+\sqrt{2(1-r_x^2)-(\|\vec{r}\|_1-r_x)^2}\right) & \text{if}\quad f(\vec{r}) \leq 0
\end{cases} \,,
\end{equation}
where $f(\vec{r}) : = 3-\|\vec{r}\|^2_1-\sqrt{6-2\|\vec{r}\|_1^2}(\|\vec{r}\|_1-3r_x)$. For the regions $P_y$ and $P_z$ the results can be obtained by exchanging $r_x \leftrightarrow r_y$ and $r_x \leftrightarrow r_z$, respectively. 
\end{proposition}
We show in Fig.~\ref{figfid} the curved surface $f(\vec{r})=0$.

\begin{figure}
     \centering
     \begin{subfigure}[h]{0.5\textwidth}
         \includegraphics[width=\textwidth]{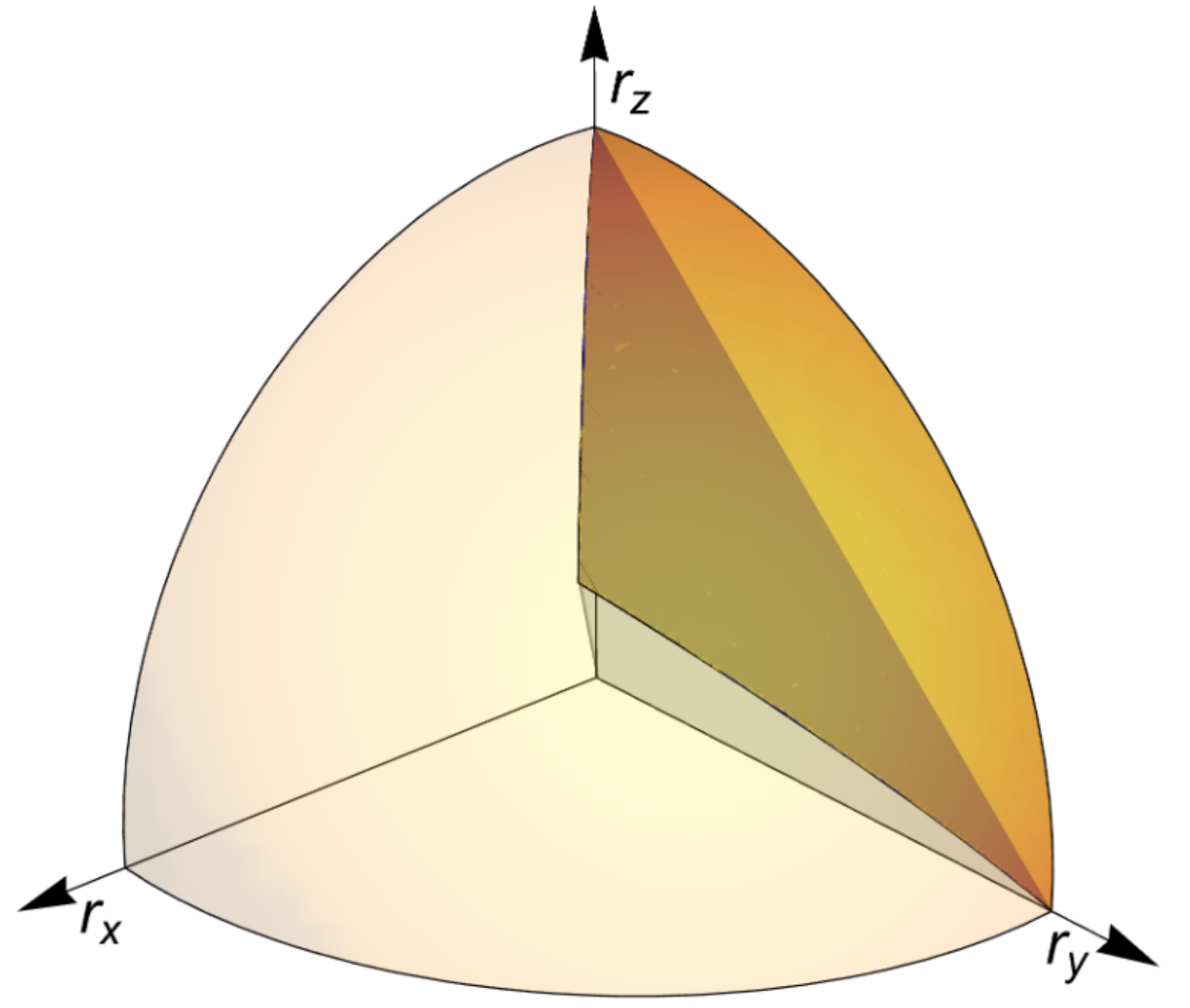}
         \caption{}
     \end{subfigure}
     \begin{subfigure}[h]{ 0.48\textwidth }
         \centering
           \includegraphics[width=\textwidth]{ 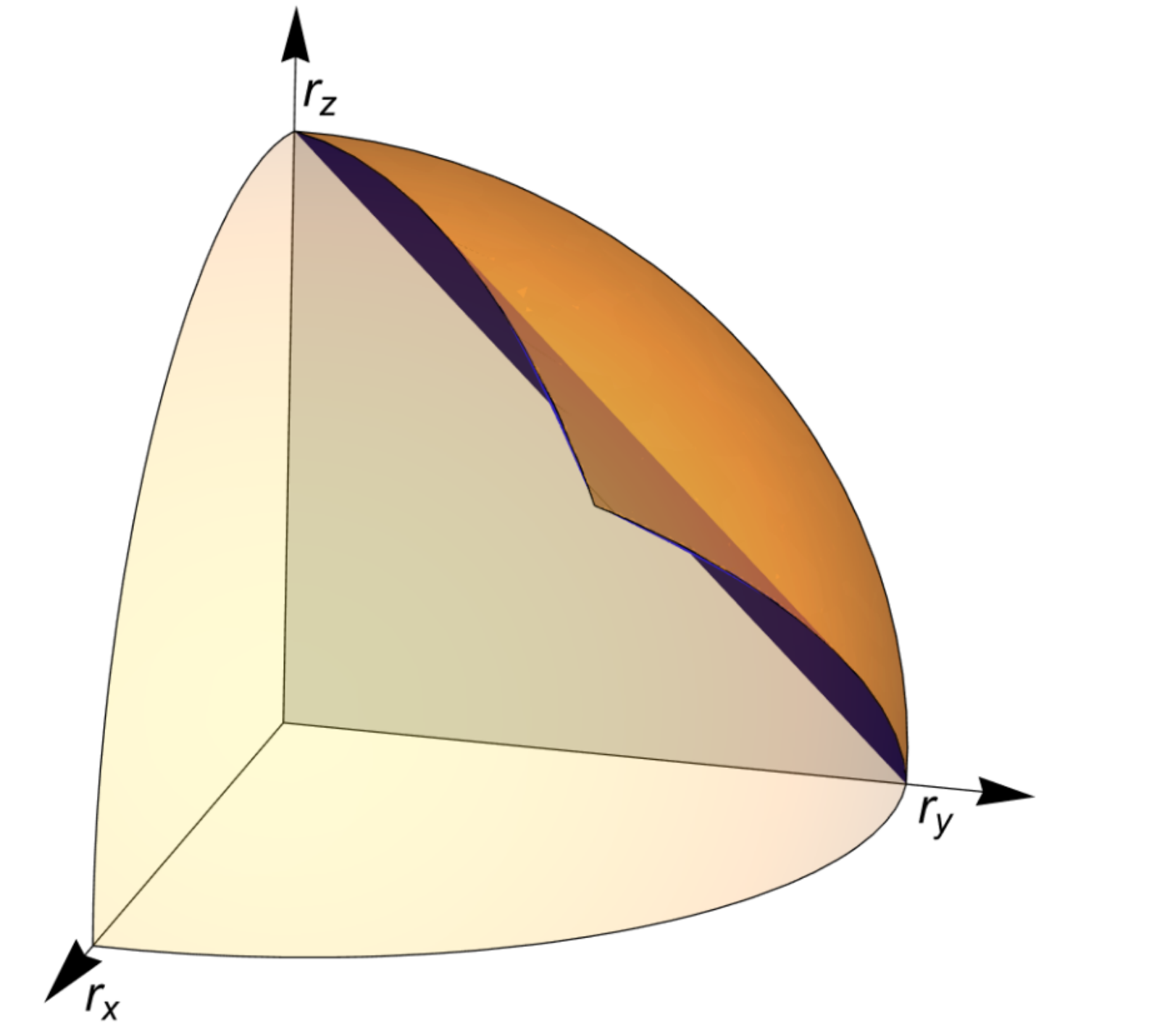}
         \caption{}
     \end{subfigure}
     \\
     \begin{subfigure}[h]{0.49\textwidth}
         \includegraphics[width=\textwidth]{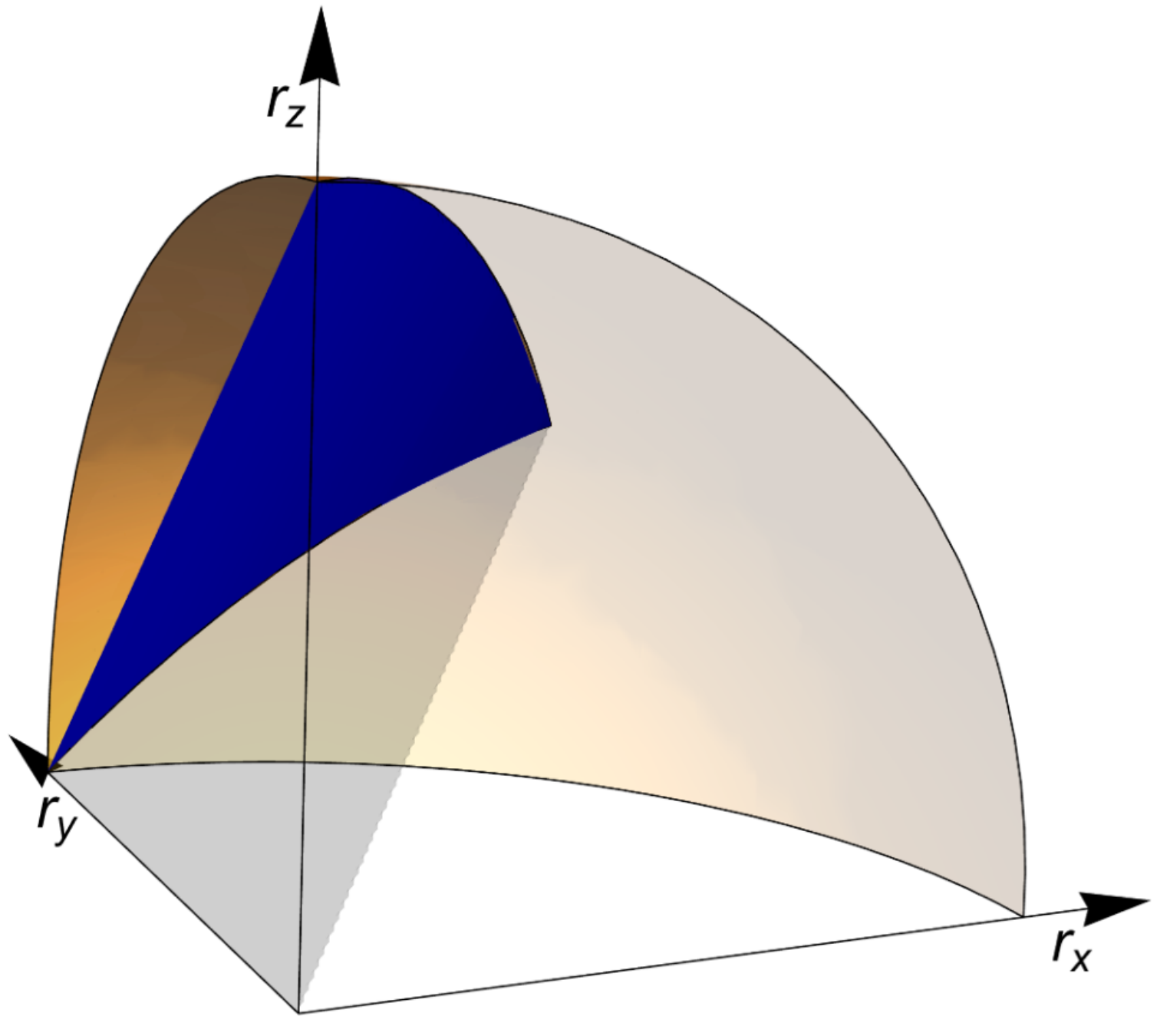}
         \caption{}
     \end{subfigure}
     \begin{subfigure}[h]{0.49\textwidth}
         \centering
         \includegraphics[width=\textwidth]{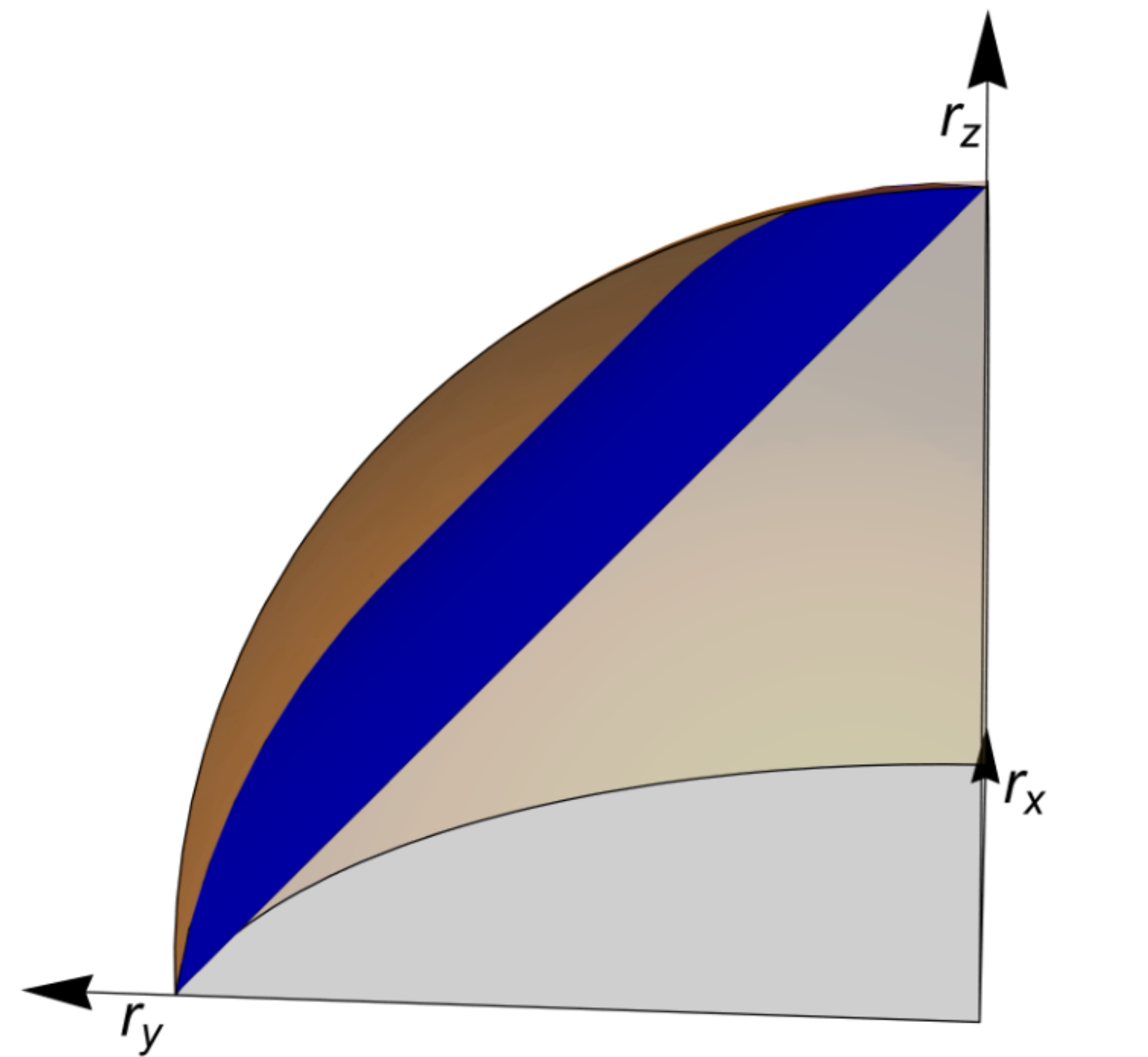}
         \caption{}
     \end{subfigure}
\caption{The figure represents four different angles of the curved surface $f(\vec{r})=0$ defined in Proposition~\ref{Fidelity} for the positive octant of the Bloch sphere. The shaded black area is the region $P_x$ defined in equation~\eqref{regions}. In the region closer to the origin it holds $f(\vec{r})\geq 0$. The opposite condition holds on the other side of the curved surface.}
\label{figfid}
\end{figure}

\begin{proof}
The fidelity between two qubits $\rho$ and $\sigma$ with Bloch vectors $\vec{r}$ and $\vec{s}$ is~\cite{ma2008geometric}
\begin{equation}
F(\rho, \sigma) = \frac{1}{2}\left(1+\vec{r}\cdot \vec{s} + \sqrt{(1-|\vec{r}|^2)(1-|\vec{s}|^2)}\right) \,.
\end{equation}
We consider the surface with constant $F(\rho, \sigma)$. We set $k=2F(\rho,\sigma)-1$. We define $h(\vec{s}) = \vec{r}\cdot\vec{s}$, $m(\vec{s}) = 1-|\vec{s}|^2$ and $p=1-|\vec{r}|^2$. Using the explicit form of the fidelity we obtain 
\begin{equation}
\frac{k-h(\vec{s})}{\sqrt{p}} = \sqrt{m(\vec{s})} \,.
\end{equation}
We then take the square and rearrange the terms and
use the explicit form of the coefficients $h(\vec{s})$ and $m(\vec{s})$ to get 
\begin{align}
&s_x^2(r_x^2+p)+s_y^2(r_y^2+p)+s_z^2(r_z^2+p)\\
&\quad +2s_xs_yr_xr_y+2s_xs_zr_xr_z+2s_ys_zr_yr_z-2ks_xr_x-2ks_yr_y-2kr_zs_z+k^2-p=0 \,.
\end{align}
The above equation is a rotated and shifted ellipsoid. To calculate the value of the stabilizer fidelity we consider the intersection with the stabilizer plane $s_x+s_y+s_z=1$ and calculate the value of $k$ for which the ellipsoid intersects the plane in exactly one point. The intersection gives the equation
\begin{align}
&as_x^2+2bs_xs_y+cs_y^2+2ds_x+2fs_y+g=0 \,,
\end{align}
where 
\begin{align}
&a = (r_x-r_z)^2+2p \,,\\
&b = r_xr_y-r_xr_z-r_yr_z+r_z^2+p \,,\\
&c = (r_y-r_z)^2+2p \,,\\
&d= -r_z^2-p+r_xr_z-kr_x+kr_z \,,\\
&f = -r_z^2-p+r_yr_z-kr_y+kr_z \,,\\
&g = (k-r_z)^2 \,,
\end{align}
which is an equation of an ellipse. From standard geometry analysis, it follows that the center of the ellipse is given by $(x_k,y_k)=(b^2-ac)^{-1}(cd-bf,af-bd)$. Moreover, the semi-axis of the ellipse are zero when $af^2+cd^2+gb^2-2bdf-acg=0$. A straightforward but lengthy calculation shows that this is the case when 
\begin{equation}
k^*=\frac{1}{3}(\|\vec{r}\|_1+\sqrt{6-2\|\vec{r}\|_1^2}) \,.
\end{equation} 
We now distinguish two cases. The first case is the one for which the projection is inside the Bloch sphere, i.e., $x_k^* \geq 0$. A calculation reveals that since $\rho \in P_x$ it holds that $x_k \leq y_k$.  Using the value of $k^*$, the condition $x_{k}^*\geq 0$ gives
\begin{equation}
3-\|\vec{r}\|^2_1-\sqrt{6-2\|\vec{r}\|_1^2}(\|\vec{r}\|_1-3r_x) \geq 0 \,.
\end{equation}  

The second case is the one for which using the above value of $k^*$, the projection onto the plane lies outside the Bloch sphere, i.e., $x_k^*\leq 0$. In this case, since the optimizer is on the $y$-$z$ plane, we calculate the intersection with the stabilizer straight line $s_z=1-s_y$, $s_x=0$. Hence, we need to solve the equation
\begin{equation}
s_y^2(r_y^2+p)+(1-s_y)^2(r_z^2+p)+2s_y(1-s_y)r_yr_z-2s_ykr_y-2(1-s_y)kr_z+k^2-p=0 \,.
\end{equation}
It is easy to check that the above equation has only one allowed solution for 
\begin{equation}
k^* = \frac{1}{2}\left(r_y+r_z+\sqrt{2-2r_x^2-(r_y+r_z)^2}\right) \,.
\end{equation}

\end{proof}

\section{Examples of two and three-qubit states belonging to the additivity classes and their computation}
\label{specific states}
In this section, we show that for some specific states, the optimizer of any monotone based on an $\alpha$-$z$ R\'enyi divergence commutes with the state itself. 

In some cases which involve pure states subject to depolarized noise $\rho= \Delta_p(\ket{\psi}\!\!\bra{\psi})$, the optimizer can be inferred from the knowledge of the solution of the related linear optimization $\mfF_0(\ket{\psi}\!\!\bra{\psi})$. This is typically the case when $\ketbra{\psi}{\psi}$ has some symmetry.
\begin{lemma}
\label{ansatz commuting}
Let $(\alpha,z) \in \mathcal{D}$ and let $\rho$ be a quantum state of dimension $d$ such that $\rho= \Delta_p(\ket{\psi}\!\!\bra{\psi}) \notin \F$ for some $p \geq 0$ and a pure state $\ket{\psi}\!\!\bra{\psi}$. If $\tau = \Delta_t(\ket{\psi}\!\!\bra{\psi}) \in \F$ for $t=(d\mfF_0(\ket{\psi}\!\!\bra{\psi})-1)/(d-1)$, then we have that $\tau \in \argmin_{\sigma \in \F} D_{\alpha,z}(\rho \| \sigma)$. 
\end{lemma}
\begin{proof}
We need to check that the state $\tau = \Delta_t(\ket{\psi}\!\!\bra{\psi})$ satisfies the condition of Corollary~\ref{commuting}
\begin{equation}
\max \limits_{\sigma \in \F}\Tr(\sigma \rho^\alpha \tau^{-\alpha})= \Tr(\rho^\alpha \tau^{1-\alpha}) \,.
\end{equation}
If we plug in the explicit form of $\tau$ and $\rho$, we get the condition
\begin{equation}
(a-b)\mfF_0(\ket{\psi}\!\!\bra{\psi})+b=a'-b'+db' \,,
\end{equation}
where 
\begin{align}
a = \left(\frac{1+(d-1)p}{1+(d-1)t}\right)^\alpha, \quad 
b = \left(\frac{1-p}{1-t}\right)^\alpha, \quad
a'=\frac{1-t}{d}a+at, \quad
b'= \frac{1-t}{d}b  \,. 
\end{align}
Hence, the condition is equivalent to 
\begin{equation}
(a-b)\left(\mfF_0(\ket{\psi}\!\!\bra{\psi}) - \frac{1-t}{d}-t\right) = 0 \,.
\end{equation}
Since $a>b$ because we assume $\rho$ is magic and hence $p>t$, the above product is zero only if $t=(d\mfF_0(\ket{\psi}\!\!\bra{\psi})-1)/(d-1)$ which proves the theorem.
\end{proof}
We point out that the previous lemma can be applied only to specific states since for general states $\tau = \Delta_t(\ket{\psi}\!\!\bra{\psi}) \notin \F$ for $t=(d\mfF_0(\ket{\psi}\!\!\bra{\psi})-1)/(d-1)$. We give a few examples below. We calculate the value of the $\alpha$-$z$ R\'enyi relative entropies of magic for all the values of the parameters. As discussed in Section~\ref{Magic monotones}, the value of the stabilizer fidelity, the relative entropy of magic, and the generalized log-robustness of magic can be obtained by taking the limits $\alpha \rightarrow 1/2$, $\alpha \rightarrow 1$ and $\alpha \rightarrow \infty$, respectively.
\subsection{T, H and F states}
\label{T, H and F states}
The $T,H$, and $F$ states are the pure states located in the symmetry axes of the stabilizer octahedron in the positive octant of the Bloch sphere. They are defined as
\begin{align}
     \ketbra{H}{H}=\frac{1}{2}\left(1+\frac{X+Z}{\sqrt{2}}\right)\,, \;  \ketbra{T}{T}=\frac{1}{2}\left(1+\frac{X+Y}{\sqrt{2}}\right)\, , \;
   \ketbra{F}{F}=\frac{1}{2}\left(1+\frac{X+Y+Z}{\sqrt{3}}\right) \,.
\end{align}
We proved in Section~\ref{closed forms} that for the stabilizer fidelity and the generalized robustness, the optimizer of the latter states lies at the intersection between the line joining the states with the center of the Bloch sphere and the stabilizer octahedron. This is true even in the presence of depolarizing noise. As a result, for the states belonging to the symmetry axes of the octahedron, the optimizer commutes with the state itself.  Here, we extend the latter result to include all the $\alpha$-$z$ R\'enyi relative entropy of magic, which notably, includes the relative entropy of magic. These states are the only single-qubit states that commute with their optimizer.

\begin{proposition}
\label{cor:noisy T}
Let $p \geq 1/\sqrt{2}$. We have
\begin{align}
\label{analytics}
&\mathfrak{D}_{\alpha,z}(\Delta_p(\ketbra{H}{H})) =
 \mathfrak{D}_{\alpha,z}(\Delta_p(\ketbra{T}{T})) \\
 & \quad \qquad \qquad \qquad \quad= 
 \frac{1}{\alpha-1} \log\left(\frac{\left(1+p\right)^\alpha}{2} \left(1+\frac{1}{\sqrt{2}}\right)^{1-\alpha} + \frac{\left(1-p\right)^\alpha}{2} \left(1-\frac{1}{\sqrt{2}}\right)^{1-\alpha} \right) \,.
\end{align}
Let $p \geq 1/\sqrt{3}$. We have
\begin{align}
&\mathfrak{D}_{\alpha,z}(\Delta_p(\ketbra{F}{F}))=
 \frac{1}{\alpha-1} \log\left(\frac{\left(1+p\right)^\alpha}{2} \left(1+\frac{1}{\sqrt{3}}\right)^{1-\alpha} + \frac{\left(1-p\right)^\alpha}{2} \left(1-\frac{1}{\sqrt{3}}\right)^{1-\alpha} \right) \,.
\end{align}
Moreover, for the above states, the optimizer commutes with the input state. 
\end{proposition}
\begin{proof}
Since the pure single-qubit stabilizer states are the eigenvectors of the Pauli matrices, by explicit computation, it is easy to prove that for a single-qubit state $\rho$ it holds $\mfF_0(\rho) = (1+\max_i|r_i|)/2$. Here, we denoted with $\vec{r}$ the Bloch vector of $\rho$. Hence, we have that
\begin{equation}
\mfF_0(\ketbra{H}{H}) = \mfF_0(\ketbra{T}{T}) = \frac{1}{2}\left(1+\frac{1}{\sqrt{2}}\right) \, \quad \mfF_0(\ketbra{F}{F}) = \frac{1}{2}\left(1+\frac{1}{\sqrt{3}}\right) \,.
\end{equation}
Setting $t_{\ketbra{\psi}{\psi}}=(d\mfF_0(\ket{\psi}\!\!\bra{\psi})-1)/(d-1)$, it is straightforward to prove that for these specific states, $\Delta_{t_{\ketbra{\psi}{\psi}}}(\ketbra{\psi}{\psi})$ is a stabilizer state. Here, we set $\psi=H,T$ and $F$. Hence, we can apply Lemma~\ref{ansatz commuting}. The proposition follows by explicit computation. 
\end{proof}

Outside the above ranges of $p$, the input state is a stabilizer state and hence the monotones evaluate to zero.

\begin{remark}
The depolarized $H,T$ and $F$ states are the only single-qubit magic states that commute with their optimizer. To see this, let us consider an arbitrary magic single-qubit state $\rho$. We can always write $\rho=\Delta_p(\ketbra{\psi}{\psi})$ for some value $p$. Let $\vec{r}$ be the Bloch vector of $\ketbra{\psi}{\psi}$. We assume without loss of generality that $r_z \geq r_y \geq r_x \geq 0$.  For the optimizer $\tau$ to commute with $\rho$, it must also be that $ \tau=\Delta_t(\ketbra{\psi}{\psi})$ for some $t$. Similarly to Lemma~\ref{ansatz commuting}, the necessary conditions of Corollary~\ref{commuting} yield $t=r_z$. Since $\tau$ must be a stabilizer state it must be that $r_z=t \leq 1/(r_x+r_y+r_z)$. Note that $r^2_x+r^2_y+r^2_z=1$. Then, the inequality is satisfied only for (1) $r_z=1$, (2) $r_x=r_y, r_z=0$, (3) $r_x=r_y=r_z$.
Therefore, the optimizer commutes with the input state only if it belongs to a symmetry axis of the stabilizer octaedron.
\end{remark}

\subsection{Toffoli state}
The Toffoli state is 
\begin{equation}
\ket{\text{Toffoli}}=\frac{\ket{000}+\ket{100}+\ket{010}+\ket{111}}{2} \,.
\end{equation}
The Toffoli state is Clifford equivalent to the $\ket{CCZ}$ state~\cite{howard2017application}. 
We now show that for the Toffoli state, we have that $\mfF_0(\ket{\text{Toffoli}}\!\!\bra{\text{Toffoli}}) = 9/16$.  Note first that the maximization is always achieved by a pure stabilizer state due to the linearity of the trace. 

Any $n$-qubit pure stabilizer state is a simultaneous $+1$ or $-1$ eigenstate of $n$ commuting independent Pauli operators and their arbitrary product ($2^n$ Pauli operators in total).
  This means that, for an arbitrary $n$-qubit pure stabilizer state $\ket{\phi}$, there exists a set of commuting $n$-qubit independent Pauli operators $\{P_k\}_{k=1}^n$ and eigenvalues $\lambda_k\in\{+1,-1\}$ such that 
\begin{align}
   \ketbra{\phi}{\phi}=\prod_{k=1}^n \frac{\mbI+{ \lambda_k}P_k}{2} = \frac{1}{2^n} \sum_{Q\in\mS}  \lambda_Q Q \,,
   \label{eq:stabilizer overlap to maximize}
\end{align}
  where $\mS$ is the set of Pauli operators generated by $\{P_i\}_{i=1}^n$, which has size $2^n$, and $\lambda_Q\in\{+1,-1\}$ that satisfies $(-1)^k \lambda_{Q_1} \lambda_{Q_2}= \lambda_Q$ for $Q_1,Q_2\in\mS$ such that $Q_1Q_2=(-1)^k Q,\ k\in\{0,1\}$.
We expand $\ketbra{\text{Toffoli}}{\text{Toffoli}} = 2^{-n} \sum_{Q=1}^{4^n} r_Q Q$ in the Pauli basis. Here, $r_Q=\Tr(Q\ketbra{\text{Toffoli}}{\text{Toffoli}})$ and since the Toffoli state is a three-qubit state, we have that $n=3$.   Noting that every non-identity Pauli operator is traceless, the terms that survive in $|\braket{\phi}{\text{Toffoli}}|^2$ are the ones with a Pauli in $\mS$. It is a straightforward calculation to check that the coefficients $r_Q$ of the Toffoli state in the Pauli basis are at most $1/2$. This gives for any pure stabilizer state $\ket{\phi}$
\begin{align}
   |\braket{\phi}{\text{Toffoli}}|^2 = \frac{1}{8}\sum_{Q\in\mS}  \lambda_Q r_Q \leq \frac{1}{8}\left(1+\frac{7}{2}\right) = \frac{9}{16} \,.
\end{align}
An easy calculation reveals that the above bound is achieved by the stabilizer state $\ket{\phi}$ with generators $\{\mathds{1} \otimes \mathds{1} \otimes \mathds{1}, \mathds{1} \otimes \mathds{1} \otimes X,\mathds{1} \otimes X \otimes \mathds{1},Z \otimes \mathds{1} \otimes \mathds{1}\}$. The latter result has been obtained in~\cite{bravyi2019simulation}.

From Lemma~\ref{ansatz commuting} we get the  stabilizer state $\tau=\frac{1}{2}\ket{\text{Toffoli}}\!\!\bra{\text{Toffoli}} + \frac{1}{2}\frac{\mathds{I}}{8}$. By direct computation, we obtain
\begin{proposition}
\label{cor:noisy Toffoli}
Let $p \geq 1/2$. We have
\begin{align}
&\mathfrak{D}_{\alpha,z}(\Delta_p(\ketbra{\textup{Toffoli}}{\textup{Toffoli}}))=
 \frac{1}{\alpha-1} \log\left(2^{\alpha-4}(7(1-p)^\alpha+9^{1-\alpha}(1+7p)^\alpha) \right) \,.
\end{align} 
Moreover, the optimizer is a depolarized Toffoli state. 
\end{proposition}

\subsection{Hoggar state}
The Hoggar state is defined as~\cite{hoggar199864}
\begin{equation}
\ket{\text{Hog}}=\frac{(1+i)\ket{000}+\ket{010}+\ket{110}+\ket{100}+\ket{101}}{\sqrt{6}} \,.
\end{equation}
For the Hoggar state we have that $\mfF_0(\ket{\text{Hog}}\!\!\bra{\text{Hog}}) = 2^{-\mathfrak{D}_{\min}(\ket{\text{Hog}}\!\!\bra{\text{Hog}})} = 5/12$~\cite{Takagi2022one-shot}. Using the previous lemma we get the stabilizer optimizer $\tau=\frac{1}{3}\ket{\text{Hog}}\!\!\bra{\text{Hog}} + \frac{2}{3}\frac{\mathds{I}}{8}$.

\begin{proposition}
\label{cor:noisy Hoggar}
Let $p \geq 1/3$. We have
\begin{align}
&\mathfrak{D}_{\alpha,z}(\Delta_p(\ketbra{\textup{Hog}}{\textup{Hog}}))=
 \frac{1}{\alpha-1} \log\left(\frac{3^{\alpha-1}}{2^{\alpha+2}}(7(1-p)^\alpha+5^{1-\alpha}(1+7p)^\alpha) \right) \,.
\end{align} 
Moreover, the optimizer is a depolarized Hoggar state. 
\end{proposition}

\subsection{CS state}
The $\ket{CS}=CS|+\rangle^{\otimes 2}$ reads
\begin{equation}
\ket{CS}=\frac{1}{2}(\ket{00}+\ket{01}+\ket{10}+i\ket{11}) \,.
\end{equation}

It is a straightforward calculation to check that the coefficients $r_Q$ of the CS-state in the Pauli basis are at most $1/2$. This gives for any pure stabilizer state $\ket{\phi}$
\begin{align}
   |\braket{\phi}{CS}|^2 = \frac{1}{4}\sum_{Q\in\mS}  \lambda_Q r_Q \leq \frac{1}{4}\left(1+\frac{3}{2}\right) = \frac{5}{8} \,.
\end{align}
An easy calculation reveals that the above bound is achieved by the stabilizer state $\ket{\phi}$ with generators $\{\mathds{1} \otimes \mathds{1},- \mathds{1} \otimes Y, X \otimes \mathds{1}\}$.

From Lemma~\ref{ansatz commuting} we get the stabilizer state $\tau=\frac{1}{2}\ket{CS}\!\!\bra{CS} + \frac{1}{2}\frac{\mathds{I}}{4}$. By direct calculation, we get 
\begin{proposition}
\label{cor:noisy CS}
Let $p \geq 1/2$. We have
\begin{align}
&\mathfrak{D}_{\alpha,z}(\Delta_p(\ketbra{\text{CS}}{\text{CS}}))=
 \frac{1}{\alpha-1} \log\left(2^{\alpha-3}(3(1-p)^\alpha+5^{1-\alpha}(1+3p)^\alpha) \right) \,.
\end{align} 
Moreover, the optimizer is a depolarized CS state.
\end{proposition}

\section{No additivity for any monotone based on a quantum relative entropy for general odd-dimensional states}
\label{counterexample}
In~\cite{veitch2014resource} the authors provided a counterexample to the additivity of the relative entropy of magic for qutrit states. We now show that the latter statement can be generalized to include any monotone based on a quantum relative entropy. We call a monotone based on a quantum relative entropy any monotone of the form $\min_{\sigma \in \mathcal{F}} \mathbb{D}(\rho \| \sigma)$ where $\mathbb{D}$ is a quantum relative entropy~\cite{gour2020optimal}.  We have
\begin{proposition}
Any monotone based on a quantum divergence is not additive for two pure qutrit strange states.
\end{proposition}
\begin{proof}
The relative entropy is not additive for the tensor product of two qutrit strange states~\cite{veitch2014resource}. Moreover, the value of the relative entropy of magic is equal for both strange states and tensor products of strange states to the value of the $\mathfrak{D}_{\min}$ and $\mathfrak{D}_{\max}$~\cite[Proposition 14]{Takagi2022one-shot}. Since any monotone based on a quantum relative entropy is contained between $\mathfrak{D}_{\min}$ and $\mathfrak{D}_{\max}$~\cite{gour2020optimal}, it follows that any monotone based on a quantum divergence is not additive for general states.
\end{proof}
Note that, since the counterexample involves pure states, it extends to any monotone defined through convex-roof construction. 

\section{A possible direction to prove the additivity of the $\alpha$-$z$ R\'enyi relative entropies of magic in the range $|(1-\alpha)|/z=1$ for all two and three-qubit states}
\label{Further directions}
In this section, we propose a direction to prove additivity for all two and three-qubit states for the $\alpha$-$z$ R\'enyi relative entropies of magic in the range $|(1-\alpha)|/z=1$. To do so, we generalize to mixed states the stabilizer-aligned condition introduced in~\cite[Definition 9]{bravyi2019simulation} for pure states. 

The 2-norm of a matrix $A$ is defined as $\|A \|_2 := (\Trm(A^\dagger A))^\frac{1}{2}$. We call $S_{n,m}$ the set of rank-$2^m$ stabilizer projectors on $n$ qubits. 
\begin{definition}
    For any n-qubit state $\rho$, define
    \begin{equation}
        \mathfrak{F}_m(\rho)=2^{-m/2} \max \limits_{\Pi_m \in S_{n,m}} \|\Pi_m \rho \Pi_m\|_2 \qquad m=0,..,n \,.
    \end{equation}

We say that $\rho$ is stabilizer-aligned if $\mathfrak{F}_m(\rho) \leq   \mathfrak{F}_0(\rho)$ for all $m$.
\end{definition}

\begin{lemma}
\label{stabilizer aligned proof}
Let $\rho_1$ and $\rho_2$ be stabilizer-aligned. Then, it holds that $\mfF_0(\rho_1 \otimes \rho_2)= \mfF_0(\rho_1)\mfF_0(\rho_2)$. 
\end{lemma} 
\begin{proof}
Since $\mfF_0$ is super-multiplicative, we need to show that for two stabilizer-aligned states, $\mfF_0(\rho_1 \otimes \rho_2) \leq \mfF_0(\rho_1)\mfF_0(\rho_2)$. We use that for a bipartite stabilizer state $|\phi\rangle$, there exists a local Clifford $U_A \otimes U_B$ and $d \geq 0$ such that $U_A \otimes U_B \ket{\phi}= \ket{\omega_d} = 2^{-d/2} \sum_{\delta=1}^{2^d}\ket{\delta}\ket{\delta}$~\cite[Theorem 5]{bravyi2006ghz}. Since local Clifford unitary does not change $\mfF_0$, we absorb the unitaries $U_A$ and $U_B$ into the states $\rho_1$ and $\rho_2$. We use the spectral decomposition into (not normalized) orthogonal states $\rho_1 = \sum_i \ketbra{\phi_i}{\phi_i}$ and $\rho_2 =  \sum_j \ketbra{\psi_j}{\psi_j}$. We have
\begin{align}
\mfF_0(\rho_1 \otimes \rho_2) 
 = \sum_{ij} \braket{\omega_d}{\phi_i \psi_j}\braket{\phi_i \psi_j}{\omega_d} 
 = \sum_{ij}|\braket{\phi_i \psi_j}{\omega_d}|^2 
 = 2^{-d}\sum_{ij}\left| \sum_{\delta=1}^{2^d}\braket{\phi_i}{\delta}\braket{\psi_j}{\delta}\right|^2 \,.
\end{align}
We now define the matrices
\begin{equation}
\Phi=
\begin{pmatrix}
\braket{\phi_1}{1} & ... & \braket{\phi_{2^n}}{1} \\
\vdots & ... & \vdots \\
\braket{\phi_1}{2^d} & ... & \braket{\phi_{2^n}}{2^d} 
\end{pmatrix} \,, \qquad 
\Psi=
\begin{pmatrix}
\braket{\psi_1}{1} & ... & \braket{\psi_{2^n}}{1} \\
\vdots & ... & \vdots \\
\braket{\psi_1}{2^d} & ... & \braket{\psi_{2^n}}{2^d} 
\end{pmatrix} \,.
\end{equation}
We note that $\mfF_0(\rho_1 \otimes \rho_2) = 2^{-d}\|\Phi^T \Psi \|_2^2 $. We now use a variation of Holder's inequality~\cite[Exercise IV.2.7]{bhatia2013matrix} which states $\|\Phi^T \Psi \|_2 \leq \|\Phi^T \|_4 \|\Psi \|_4$. An explicit calculation shows that $\|\Phi^T \|_4 = (\Tr(\Phi^\dagger \Phi \Phi^\dagger \Phi))^\frac{1}{4} = \Tr(\rho_1 \Pi_d \rho_1 \Pi_d)^\frac{1}{4}$ and, analogously $\|\Psi \|_4 = \Tr(\rho_2 \Pi_d \rho_2 \Pi_d)^\frac{1}{4}$ where $\Pi_d = \sum_{\delta=1}^{2^d}\ketbra{\delta}{\delta}$. If both states are stabilizer-aligned, we therefore get
\begin{align}
\mfF_0(\rho_1 \otimes \rho_2) &= 2^{-d}\|\Phi^T \Psi \|_2^2  \\
& \leq 2^{-d} \Tr(\rho_1 \Pi_d \rho_1 \Pi_d)^\frac{1}{2} \Tr(\rho_2 \Pi_d \rho_2 \Pi_d)^\frac{1}{2}\\
& \leq \mfF_0(\rho_1) \mfF_0(\rho_2) \,,
\end{align}
which proves the lemma. 
\end{proof}
We now show that single-qubit states are stabilizer-aligned.
\begin{lemma}
All single-qubit states are stabilizer-aligned.
\end{lemma}
\begin{proof}
We can write any qubit in the Bloch representation $\rho = (1/2)(1+\sum _p c_p P)$, where the sum runs over the Pauli matrices.  It is easy to check that $\mfF_1(\rho)= 2^{-1/2}\sqrt{\Tr(\rho^2)} = (1/2)\sqrt{(1+\sum_p c_p^2)}$ since the only projector of rank-two is the identity. Moreover, $\mfF_0(\rho) = (1/2)(1+|c_{p^*}|)$, where $c_{p^*} = \max_p c_p$. The condition $\mfF_0(\rho) \geq \mfF_1(\rho)$ then reads $c_{p^*}^2+2 |c_{p^*}| \geq \sum_p c_p^2$. Using that $|c_{p^*}|\leq 1$ and that the maximum is greater than the average, we get $|c_{p^*}| \geq c_{p^*}^2 \geq  \sum_p c_p^2/3$. It then follows that $c_{p^*}^2+2 |c_{p^*}| \geq (1/3+2/3) \sum_p c_p^2 = \sum_p c_p^2$ which proves the inequality. 
\end{proof}
Following similar arguments that lead to Theorem~\ref{Additivity in the line}, the above results imply the additivity of the $\alpha$-$z$ R\'enyi relative entropies of magic in the range $|(1-\alpha)|/z=1$ for two single-qubit states. It is an open problem whether the proof of Lemma~\ref{stabilizer aligned proof} can be generalized to products of more than two stabilizer-aligned states. Moreover, it is still an open question whether all the two and three-qubit states are stabilizer-aligned according to the above definition. If these two questions could be resolved in the affirmative, additivity of the $\alpha$-$z$ R\'enyi relative entropies of magic in the range $|(1-\alpha)|/z=1$ would follow for any number of any two and three-qubit states.

\section*{ACKNOWLEDGMENTS}
We thank Bartosz Regula for pointing us out some relevant references and for fruitful discussions about monotonicity of resource monotones under probabilistic transformations. 
We also thank Kohdai Kuroiwa for bringing our attention to a typographical error in Proposition~\ref{Fidelity} in the previous version of the manuscript. 
RR and MT are supported by the National Research Foundation, Singapore, and A*STAR under its CQT Bridging Grant.  RT was supported by the Lee Kuan Yew Postdoctoral Fellowship at Nanyang Technological University Singapore, JSPS KAKENHI Grant Number JP23K19028, and JST, CREST Grant Number JPMJCR23I3, Japan. 
MT is also supported by the National Research Foundation, Singapore, and A*STAR under its Quantum Engineering
Programme (NRF2021-QEP2-01-P06).

\bibliographystyle{quantum}
\bibliography{my}

\newpage

\appendix
\section*{Appendices}
\addcontentsline{toc}{section}{Appendices}
\renewcommand{\thesubsection}{\Alph{subsection}}

\subsection{Necessary and sufficient conditions for the optimizer}
\label{necessary and sufficient}
In this appendix, we review the necessary and sufficient conditions for the optimizer(s) of the $\alpha$-$z$ R\'enyi relative entropies of magic derived in~\cite{rubboli2024new}. The results are very general and hold for any resource theory. In this manuscript, we consider $\mathcal{F}=\F$. 
Let $\rho,\tau \in \state$. We define
\begin{align}
\label{equation problem}
&\Xi_{\alpha,z}(\rho,\tau) :=  \begin{dcases}
\chi_{\alpha,1-\alpha}(\rho,\tau) &  \textup{if} \; z=1-\alpha \\
\tau^{-1}  \chi_{\alpha,\alpha-1}(\rho, \tau) \tau^{-1}  &  \textup{if} \; z=\alpha-1 \\
K_{\alpha,z} \int_0^\infty \frac{\chi_{\alpha,z}(\rho,\tau)}{(\tau + t)^2} t^\frac{1-\alpha}{z} \d t & \textup{if} \; |(1-\alpha)/z| \neq 1
\end{dcases} \,,
\end{align}
where $K_{\alpha,z}:=\sinc\left(\pi\frac{{1-\alpha}}{z}\right)$ is a positive constant and we defined the positive operator $\chi_{\alpha,z}(\rho,\tau):=  \rho^\frac{\alpha}{2z}( \rho^{\frac{\alpha}{2z}}\tau^{\frac{1-\alpha}{z}} \rho^{\frac{\alpha}{2z}})^{z-1} \rho^\frac{\alpha}{2z}$.

Moreover, we define
\begin{align}
\label{sets}
S_{\alpha,z}(\rho) =  \begin{cases} 
\left\{ \sigma \in \state : \supp(\rho) = \supp\!\left(\Pi(\rho)\sigma \Pi(\rho)\right) \right\}  &  \textup{if} \; (1-\alpha)/z = 1\\
\left\{ \sigma \in \state : \supp(\rho) \subseteq \supp(\sigma)  \right\} &  \textup{otherwise} 
\end{cases} \, ,
\end{align} 

We have the following result~\cite{rubboli2024new}
\begin{theorem}
\label{main Theorem}
Let $\rho$ be a quantum state and $(\alpha,z) \in \mathcal{D}$. Then $\tau \in \argmin_{\sigma \in \mathcal{F}} D_{\alpha,z}(\rho \| \sigma)$ if and only if $\tau \in S_{\alpha,z}(\rho)$ and $\textup{Tr}(\sigma\, \Xi_{\alpha,z}(\rho,\tau)) \leq  Q_{\alpha,z}(\rho \|\tau)$ for all $\sigma \in \mathcal{F}$.
\end{theorem}

In this case, the above theorem considerably simplifies. 

\begin{corollary}
\label{commuting}
Let $\rho$ be a quantum state and $(\alpha,z) \in \mathcal{D}$. Then, a state $\tau$ satisfying $[\rho,\tau]= 0$ belongs to the set $\tau \in \argmin_{\sigma \in \mathcal{F}} D_{\alpha,z}(\rho \| \sigma)$ if and only if $\supp(\rho)\subseteq \supp(\tau)$ and $\textup{Tr}(\sigma \Xi_{\alpha}(\rho,\tau)) \leq  Q_\alpha(\rho \|\tau)$ for all $\sigma \in \mathcal{F}$ where $\Xi_{\alpha}(\rho,\tau) = \rho^\alpha \tau^{-\alpha}$ and $Q_\alpha(\rho \|\tau) = \textup{Tr}(\rho^\alpha \tau^{1-\alpha})$. 
\end{corollary}

\subsection{Properties of the resource monotones based on the $\alpha$-$z$ R\'enyi relative entropies}
\label{distillation lower bound}
In this Appendix, we discuss some of the properties satisfied by the monotones based on the $\alpha$-$z$ R\'enyi relative entropies. These properties have already been discussed extensively in entanglement theory~\cite{vedral1998entanglement,wei2003geometric} for some specific values of the parameters. Moreover, the generalized robustness case has been discussed in~\cite{regula2017convex} for any convex resource theories. Here, we give a more general proof that holds for any convex resource theory and all the values of the parameters $(\alpha,z) \in \mathcal{D}$. In this case, the same definitions given in the main text hold with the set of stabilizer states $\F$ replaced by a general convex set of free states $\mathcal{F}$. We first prove that the monotones are convex or concave, depending on the value of the parameter.

\begin{lemma}
Let $(\alpha,z) \in \mathcal{D}$, $\{\rho_i\}_i$ a set of states and $\{p_i\}_i$ a set of weights, i.e., $p_i \geq 0$ and $\sum_i p_i=1$. Then, we have
\begin{align}
&\mathcal{Q}_{\alpha,z}^\frac{1}{\alpha}\Big(\sum \nolimits_i  p_i\rho_i\Big) \geq \sum \nolimits_i p_i \mathcal{Q}_{\alpha,z}^\frac{1}{\alpha}(\rho_i)\,, \quad\quad \alpha < 1 \\
&\mathcal{Q}_{\alpha,z}^\frac{1}{\alpha}\Big(\sum \nolimits_i  p_i\rho_i\Big) \leq \sum \nolimits_i p_i \mathcal{Q}_{\alpha,z}^\frac{1}{\alpha}(\rho_i)\,, \quad\quad \alpha > 1  \\
&\mathfrak{D}\Big(\sum \nolimits_i  p_i\rho_i\Big)\leq \sum \nolimits_i p_i \mathfrak{D}(\rho_i) \,.
\end{align}
\end{lemma}
\begin{proof}
We first consider case $\alpha>1$. We denote with $\tau_i$ an optimizer of $\rho_i$, i.e a state such that $\mathcal{Q}_{\alpha,z}(\rho_i) = Q_{\alpha,z}(\rho_i\| \tau_i)$. We define the probability vector
\begin{equation}
t_k = \frac{p_k Q^\frac{1}{\alpha}_{\alpha,z}(\rho_k \| \tau_k)}{\sum_i p_i Q^\frac{1}{\alpha}_{\alpha,z}(\rho_i \| \tau_i)} \,.
\end{equation}
We then have
\begin{align}
\label{free}
    \mathcal{Q}_{\alpha,z}\Big(\sum \nolimits_i  p_i\rho_i\Big) &\leq Q_{\alpha,z} \left(\sum \nolimits_i  p_i\rho_i \middle\| \sum \nolimits_i  t_i\tau_i \right ) \\
    \label{DPI}
    &\leq Q_{\alpha,z} \left(\sum \nolimits_i  p_i\rho_i \otimes \ketbra{i}{i} \middle\| \sum \nolimits_i  t_i\tau_i \otimes \ketbra{i}{i} \right ) \\
    & = \sum \nolimits_i  p_i^\alpha t_i^{1-\alpha}  Q_{\alpha,z} \left(\rho_i \middle\| \tau_i  \right ) \\
    &= \left(\sum \nolimits_i p_i \mathcal{Q}_{\alpha,z}^\frac{1}{\alpha}(\rho_i)\right)^\alpha \,.
\end{align}
In the first inequality~\eqref{free} we used that a mixture of free states is free and~\eqref{DPI} follows for the DPI under partial trace. The last two steps follow by direct computation. The case $\alpha<1$ is similar and the result for the relative entropy follows from joint convexity of the relative entropy (see, e.g.,~\cite{baumgratz2014quantifying}).
\end{proof}
We now discuss the strong monotonicity property (see e.g.,~\cite{vedral1998entanglement,seddon2019quantifying,baumgratz2014quantifying,regula2017convex}). 
An instrument is a collection of complete positive trace-nonincreasing operations $\{\Phi_i\}$ such that the overall transformation $\sum_i \Phi_i$ is trace-preserving. We say that $\{\Phi_i\}$ is a free probabilistic instrument if each operation maps a free state into a free state up to a normalization factor, i.e., $\Phi_i(\sigma) \propto \sigma' \in \mathcal{F}$ for any $\sigma \in \mathcal{F}$. We follow the definition given in~\cite{regula2022tight}.
\begin{definition}
Let $\{\Phi_i\}$ be a free probabilistic instrument.  A resource monotone $\mathcal{M}$ is called a strong monotone if it decreases on average under the action of $\{\Phi_i\}$, i.e., it satisfies
\begin{equation}
\mathcal{M}(\rho) \geq \sum_i \Tr(\Phi_i(\rho)) \mathcal{M} \left( \frac{\Phi_i(\rho)}{\Tr(\Phi_i(\rho))} \right) \,.
\end{equation}
\end{definition}
We first prove an auxiliary lemma.
\begin{lemma}
\label{montonicity divergences}
    Let $(\alpha,z) \in \mathcal{D}$, $\rho \in \state$ and $\{ \Phi_i \}$ be an instrument. Then, we have
    \begin{align}
&Q_{\alpha,z}^\frac{1}{\alpha}(\rho\|\sigma) \leq \sum \nolimits_i p_i Q_{\alpha,z}^\frac{1}{\alpha}(\rho_i \| \sigma_i)\,, \quad\quad \alpha < 1 \\
&Q_{\alpha,z}^\frac{1}{\alpha}(\rho\|\sigma) \geq \sum \nolimits_i p_i Q_{\alpha,z}^\frac{1}{\alpha}(\rho_i \|\sigma_i)\,, \quad\quad \alpha > 1  \\
&D(\rho \| \sigma) \geq  \sum \nolimits_i p_i D(\rho_i \|\sigma_i) \,,
\end{align}
 where $p_i = \Tr(\Phi_i(\rho))$, $\rho_i = \Phi_i(\rho)/\Tr(\Phi_i(\rho))$, and $\sigma_i = \Phi_i(\sigma)/\Tr(\Phi_i(\sigma))$.
\end{lemma}
\begin{proof}
    We first consider the case $\alpha > 1$. We consider the quantum channel $\Phi(\cdot) = \sum_i\Phi_i(\cdot) \otimes \ketbra{i}{i}$. Then, the DPI inequality implies that
    \begin{align}
    \label{DPI monotones}
        Q_{\alpha,z}(\rho\|\sigma) &\geq Q_{\alpha,z}\left(\sum_i p_i \rho_i \otimes \ketbra{i}{i} \middle\| \sum_i q_i \sigma_i \otimes \ketbra{i}{i}\right) \\
        \label{classical-quantum}
        &= \sum_i p_i^\alpha q_i^{1-\alpha} Q_\alpha(\rho_i \|\sigma_i) \,,
    \end{align}
    where $p_i = \Tr(\Phi_i(\rho))$, $q_i = \Tr(\Phi_i(\sigma))$, $\rho_i = \Phi_i(\rho)/\Tr(\Phi_i(\rho))$, and $\sigma_i = \Phi_i(\sigma)/\Tr(\Phi_i(\sigma))$. We now lower bound~\eqref{classical-quantum} by taking the infimum over the coefficients $q_i$. This problem can be solved using standard Lagrange multipliers techniques. 
The solution is given by
\begin{equation}
q^*_k = \frac{p_k Q^\frac{1}{\alpha}_{\alpha,z}(\rho_k \| \sigma_k)}{\sum_i p_i Q^\frac{1}{\alpha}_{\alpha,z}(\rho_i \| \sigma_i)} \,.
\end{equation}
By substitution, we obtain 
\begin{equation}
Q_{\alpha,z}(\rho \| \sigma) \geq \left(\sum_i p_i Q^\frac{1}{\alpha}_{\alpha,z}(\rho_i\|\sigma_i)\right)^\alpha \,.
\end{equation}
The other two cases are similar.
\end{proof}

The resource monotones $\mathcal{Q}_{\alpha,z}^\frac{1}{\alpha}$ satisfy the strong monotonicity for $\alpha > 1$ and strong antimonotoncity for $\alpha < 1$.

\begin{corollary}
\label{strong monotonicity}
    Let $(\alpha,z) \in \mathcal{D}$, $\rho \in \state$ and $\{\Phi_i\}$ be a free probabilistic instrument. Then, we have
    \begin{align}
&\mathcal{Q}_{\alpha,z}^\frac{1}{\alpha} (\rho)\leq \sum \nolimits_i p_i \mathcal{Q}_{\alpha,z}^\frac{1}{\alpha}(\rho_i)\,, \quad\quad \alpha < 1 \\
&\mathcal{Q}^\frac{1}{\alpha} _{\alpha,z}(\rho) \geq \sum \nolimits_i p_i \mathcal{Q}_{\alpha,z}^\frac{1}{\alpha}(\rho_i)\,, \quad\quad \alpha > 1  \\
&\mathfrak{D}(\rho) \geq  \sum \nolimits_i p_i \mathfrak{D}(\rho_i) \,,
\end{align}
 where $p_i = \Tr(\Phi_i(\rho))$, and $\rho_i = \Phi_i(\rho)/\Tr(\Phi_i(\rho))$.
\end{corollary}
\begin{proof}
The proof is a straightforward application of Lemma~\ref{montonicity divergences}. We consider only the case $\alpha >1$ as the other cases are similar. Indeed, we have
\begin{equation}
    \mathcal{Q}^\frac{1}{\alpha}_{\alpha,z}(\rho) = Q_{\alpha,z}^\frac{1}{\alpha}(\rho\| \tau) \geq  \sum_i p_i Q_{\alpha,z}^\frac{1}{\alpha}(\rho_i\| \tau_i) \geq \sum_i p_i \mathcal{Q}_{\alpha,z}^\frac{1}{\alpha}(\rho_i) \,,
\end{equation}
where we denoted by $\tau$ the optimizer of $\rho$. Here, we used Lemma~\ref{montonicity divergences} in the first inequality and the second inequality follows from the assumption that the instrument is free. 
\end{proof}
Since $\lim_{\alpha \rightarrow \infty }\mathcal{Q}_{\alpha,\alpha}^\frac{1}{\alpha}(\rho) = \Lambda^{+}(\rho)$, where $\Lambda^{+}(\rho)$ is the generalized robustness of resource, the above result for $\alpha \rightarrow \infty$ recovers the one given in~\cite{regula2017convex}. Moreover, the result for the Umegaki relative entropy recovers the one in~\cite{vedral1998entanglement}. The result for $\alpha=1/2$ was obtained for entanglement theory in~\cite{wei2003geometric}. Finally, we note that in~\cite{shao2015fidelity}, the authors proved that the square root fidelity measure does not satisfy the strong monotonicity property. For this reason, they argue that the fidelity measure is not a good measure of coherence. Here, we show that instead, the fidelity of coherence (with no square root) does satisfy the strong monotonicity property.

\subsection{$D_{\max}$ for single-qubit states}
In this Appendix, we find an explicit form of the max-relative entropy as a function of the Bloch vectors. We denote $Q_{\max}(\rho \| \sigma):=\exp(D_{\max}(\rho \| \sigma))$. Note that if $\textup{supp}(\rho) \subseteq \textup{supp}(\sigma)$, then $Q_{\max}(\rho \| \sigma)$ is infinite.
\begin{lemma}
\label{robustness qubit}
Let $\rho,\sigma$ two single-qubit states such that $\textup{supp}(\rho) \subseteq \textup{supp}(\sigma)$ and with Bloch vectors $\vec{r}$ and $\vec{s}$, respectively. Then
\begin{align}
&Q_{\max}(\rho \| \sigma) = 
(1-|\vec{s}|^2)^{-1}\left(1-\vec{r}\cdot \vec{s}+\sqrt{(1-\vec{r}\cdot \vec{s})^2-(1-|\vec{r}|^2)(1-|\vec{s}|^2)}\right)\,.
\end{align}
\end{lemma}
\begin{proof}
We consider a state $\rho$ with full support. The result for pure states follows similarly and can also be recovered by taking the limit. We can write a qubit state $\rho$ with Bloch vector $\vec{r}$ as
\begin{equation}
\rho = \frac{1}{2}\begin{pmatrix}
1+r_z & r_x-ir_y \\
r_x-ir_y & 1-r_z
\end{pmatrix} \,.
\end{equation}
The eigenvectors of $\rho$ are
\renewcommand\arraystretch{1.5}
\begin{equation}
\ket{r_1}=\frac{1}{\sqrt{2|\vec{r}|}}
\begin{pmatrix}
\frac{r_x-ir_y}{\sqrt{|\vec{r}|-r_z}}\\
\sqrt{|\vec{r}|-r_z}
\end{pmatrix} \,,\quad 
\ket{r_2}=\frac{1}{\sqrt{2|\vec{r}|}}
\begin{pmatrix}
-\frac{r_x-ir_y}{\sqrt{|\vec{r}|+r_z}}\\
\sqrt{|\vec{r}|+r_z} \,,
\end{pmatrix}
\end{equation}
with eigenvalues $r_1=(1+|\vec{r}|)/2$ and $r_2=(1-|\vec{r}|)/2$, respectively. By explicit computation, the matrix elements with a state $\sigma$ with Bloch vector $\vec{s}$ are
\begin{align}
&c_{11}=1-c_{22}=\bra{r_1}\sigma\ket{r_1}= \frac{1}{2}\left(1+\frac{\vec{r}\cdot \vec{s}}{|\vec{r}|}\right) \,, \\
&c_{12} =  \bra{r_1}\sigma\ket{r_2}=  \frac{1}{2|\vec{r}|}\frac{1}{\sqrt{|\vec{r}|^2-r_z^2}}\big(r_z(r_xs_x+r_ys_y))-(|\vec{r}|^2-r_z^2)s_z-i|\vec{r}|(r_xs_y-r_ys_x)\big) \,.
\end{align}
Moreover, we also have $c_{21}^*=c_{12}$.
Therefore, in the eigenbasis of $\rho$ we obtain
\begin{align}
\rho^\frac{1}{2}\sigma\rho^\frac{1}{2} = \begin{pmatrix}
r_1c_{11} & (r_1r_2)^\frac{1}{2}c_{12} \\
(r_1r_2)^\frac{1}{2}c_{21} & r_2c_{22}
\end{pmatrix} \,.
\end{align}
A straightforward calculation gives
\begin{equation}
\lambda_{\pm} = \frac{1+\vec{r}\cdot\vec{s}}{4} \pm \frac{1}{2}\sqrt{\left(\frac{1+\vec{r}\cdot\vec{s}}{2}\right)^2 - \left(\frac{1-|\vec{s}|^2}{4}\right)(1-|\vec{r}|^2)} \,.
\end{equation}
We recall that $Q_{\max}(\rho \| \sigma) = \| \rho^\frac{1}{2} \sigma^{-1} \rho^\frac{1}{2} \|_{\infty}$, where the infinity norm of a positive operator is equal to its maximum eigenvalue~\cite{tomamichel2015quantum}.
To calculate  we first take the inverse of $\sigma$ to get
\begin{equation}
\sigma^{-1} = \frac{4}{1-|\vec{s}|^2} \frac{1}{2}\begin{pmatrix}
1-s_z & -s_x+is_y \\
-s_x-is_y & 1+s_z
\end{pmatrix} \,.
\end{equation}
Hence, we obtain 
\begin{equation}
\rho^\frac{1}{2}\sigma^{-1}\rho^\frac{1}{2} = \frac{4}{1-|\vec{s}|^2} \rho^\frac{1}{2}\sigma_{-}\rho^\frac{1}{2} \,,
\end{equation}
where we denoted $\sigma_{-}$ the qubit state with Bloch vector $-\vec{s}$. We then use the previous result to get that the above product has eigenvalues
\begin{equation}
\frac{4}{1-|\vec{s}|^2}\left(\frac{1-\vec{r}\cdot\vec{s}}{4} \pm \frac{1}{2} \sqrt{\left(\frac{1-\vec{r}\cdot\vec{s}}{2}\right)^2- \left(\frac{1-|\vec{s}|^2}{4}\right)(1-|\vec{r}|^2)} \right) \,.
\end{equation}
By choosing the maximum eigenvalue we get the claimed result.
\end{proof}

\subsection{Exchanging the limits with the minimization}
\label{convergence monotone}
We first show that the function $D_{\alpha,\alpha-1}(\rho \| \sigma)$ converges to $D_{\max}(\rho \| \sigma)$. The next proof is a straightforward extension of the one given in~\cite[Section C]{muller2013quantum}. 
\begin{lemma}
We have that $\lim_{\alpha \rightarrow \infty} D_{\alpha,\alpha-1}(\rho \|\sigma)=D_{\max}(\rho \|\sigma)$.
\end{lemma}
\begin{proof}
We first write
\begin{equation}
\|\sigma^{-\frac{1}{2}}\rho^\frac{\alpha}{\alpha-1}\sigma^{-\frac{1}{2}} \|_{\alpha-1} = \|\sigma^{-\frac{1}{2}}\rho^\frac{\alpha}{\alpha-1}\sigma^{-\frac{1}{2}} \|_{\alpha-1} + \|\sigma^{-\frac{1}{2}}\rho\sigma^{-\frac{1}{2}} \|_{\alpha-1} - \|\sigma^{-\frac{1}{2}}\rho\sigma^{-\frac{1}{2}} \|_{\alpha-1} \,.
\end{equation}
The reverse triangle inequality implies that
\begin{equation}
\left|\|\sigma^{-\frac{1}{2}}\rho^\frac{\alpha}{\alpha-1}\sigma^{-\frac{1}{2}} \|_{\alpha-1} - \|\sigma^{-\frac{1}{2}}\rho\sigma^{-\frac{1}{2}} \|_{\alpha-1}\right| \leq \|\sigma^{-\frac{1}{2}}\rho^\frac{\alpha}{\alpha-1}\sigma^{-\frac{1}{2}}- \sigma^{-\frac{1}{2}}\rho\sigma^{-\frac{1}{2}} \|_{\alpha-1} \,.
\end{equation}
So we have that
\begin{align}
&\lim \limits_{\alpha \rightarrow \infty}D_{\alpha,\alpha-1}(\rho \|\sigma) \\
&\quad = \lim \limits_{\alpha \rightarrow \infty} \log{\|\sigma^{-\frac{1}{2}}\rho^\frac{\alpha}{\alpha-1}\sigma^{-\frac{1}{2}}\|_{\alpha-1}} \\
& \quad \leq \log{ \left( \lim \limits_{\alpha \rightarrow \infty} \|\sigma^{-\frac{1}{2}}\rho\sigma^{-\frac{1}{2}} \|_{\alpha-1} + \lim \limits_{\alpha \rightarrow \infty} \|\sigma^{-\frac{1}{2}}\rho^\frac{\alpha}{\alpha-1}\sigma^{-\frac{1}{2}}- \sigma^{-\frac{1}{2}}\rho\sigma^{-\frac{1}{2}} \|_{\alpha-1}\right)} \\
& \quad \leq  \log{ \left( \lim \limits_{\alpha \rightarrow \infty} \|\sigma^{-\frac{1}{2}}\rho\sigma^{-\frac{1}{2}} \|_{\alpha-1} + (\text{dim}\mathcal{H}) \lim \limits_{\alpha \rightarrow \infty}  \|\sigma^{-\frac{1}{2}}\rho^\frac{\alpha}{\alpha-1}\sigma^{-\frac{1}{2}}- \sigma^{-\frac{1}{2}}\rho\sigma^{-\frac{1}{2}} \|_{\infty}\right)} \\
& \quad =  \log{ \|\sigma^{-\frac{1}{2}}\rho \sigma^{-\frac{1}{2}} \|_{\infty}} \\
& \quad = D_{\max}(\rho \|\sigma) \,.
\end{align}
Likewise,
\begin{align}
&\lim \limits_{\alpha \rightarrow \infty}D_{\alpha,\alpha-1}(\rho \|\sigma) \\
&\quad = \lim \limits_{\alpha \rightarrow \infty} \log{\|\sigma^{-\frac{1}{2}}\rho^\frac{\alpha}{\alpha-1}\sigma^{-\frac{1}{2}}\|_{\alpha-1}} \\
& \quad \geq \log{ \left( \lim \limits_{\alpha \rightarrow \infty} \|\sigma^{-\frac{1}{2}}\rho\sigma^{-\frac{1}{2}} \|_{\alpha-1} - \lim \limits_{\alpha \rightarrow \infty} \|\sigma^{-\frac{1}{2}}\rho^\frac{\alpha}{\alpha-1}\sigma^{-\frac{1}{2}}- \sigma^{-\frac{1}{2}}\rho\sigma^{-\frac{1}{2}} \|_{\alpha-1}\right)} \\
& \quad \geq \log{ \left( \lim \limits_{\alpha \rightarrow \infty} \|\sigma^{-\frac{1}{2}}\rho\sigma^{-\frac{1}{2}} \|_{\alpha-1} - (\text{dim}\mathcal{H}) \lim \limits_{\alpha \rightarrow \infty}  \|\sigma^{-\frac{1}{2}}\rho^\frac{\alpha}{\alpha-1}\sigma^{-\frac{1}{2}}- \sigma^{-\frac{1}{2}}\rho\sigma^{-\frac{1}{2}} \|_{\infty}\right)} \\
& \quad =  \log{ \|\sigma^{-\frac{1}{2}}\rho \sigma^{-\frac{1}{2}} \|_{\infty}} \\
& \quad = D_{\max}(\rho \|\sigma) \,.
\end{align}
\end{proof}

We now prove that $D_{\alpha,\alpha-1}(\rho \| \sigma)$ and $D_{\alpha,1-\alpha}(\rho \| \sigma)$ are monotonically increasing in $\alpha$. 

\begin{lemma}
\label{monotonicity}
The function $\alpha \mapsto D_{\alpha,1-\alpha}(\rho \| \sigma)$ is montonically increasing on $(0,1)$ and $\alpha \mapsto D_{\alpha,\alpha-1}(\rho \| \sigma)$ is   monotonically increasing on $(1,\infty)$.
\end{lemma}

\begin{proof}
We first start with the case  $\alpha \mapsto D_{\alpha,\alpha-1}(\rho \| \sigma)$. We denote $\beta:=\alpha-1$ (and, equivalently, derivate in $\beta$) and $A:=(\sigma^{-\frac{1}{2}}\rho^\frac{\alpha}{\alpha-1}\sigma^{-\frac{1}{2}})= (\sigma^{-\frac{1}{2}}\rho^{1+\frac{1}{\beta}}\sigma^{-\frac{1}{2}})$. To calculate the derivative, we follow similar steps to the one given to prove~\cite[Lemma III.1]{mosonyi2015two} (see also~\cite[Lemma 7]{lin2015investigating}). If the condition $\rho \ll \sigma$ is not satisfied, then the value of the monotones is $+\infty$ and the theorem holds. If the support condition is satisfied,  then, $\alpha \mapsto D_{\alpha,\alpha-1}(\rho \| \sigma)$ is differentiable on $(1,\infty)$, and  
\begin{align}
&\frac{\d}{\d \beta} D_{\alpha,\alpha-1}(\rho \| \sigma) \\
&\quad = -\frac{1}{\beta^2} \log{\Trm (A^\beta)}+\frac{1}{\beta\Trm(A^\beta)}\Trm\left(A^\beta\log{A}+\beta\sigma^{-\frac{1}{2}}\rho^{1+\frac{1}{\beta}}(\log{\rho})\sigma^{-\frac{1}{2}} A^{\beta-1} \left(-\frac{1}{\beta^2}\right)\right) \\
&\quad =\frac{1}{\beta^2 \Trm(A^\beta)}\left(\Trm(A^\beta \log{A^\beta})-\Trm(A^\beta)\log{\Trm(A^\beta)} - \Trm(\sigma^{-\frac{1}{2}}\rho^{1+\frac{1}{\beta}}(\log{\rho})\sigma^{-\frac{1}{2}} A^{\beta-1} )\right) \,.
\end{align}
We have that $f(Y^\dagger Y)Y^\dagger = Y^\dagger f(YY^\dagger)$. We set $Y = \sigma^{-\frac{1}{2}} \rho^{\frac{1}{2}+\frac{1}{2\beta}}$. We can then rewrite in the following
\begin{equation}
\Trm(\sigma^{-\frac{1}{2}}\rho^{1+\frac{1}{\beta}}(\log{\rho})\sigma^{-\frac{1}{2}} A^{\beta-1} ) = \Trm(\log{\rho}(\rho^{\frac{1}{2}+\frac{1}{2\beta}}\sigma^{-1}\rho^{\frac{1}{2}+\frac{1}{2\beta}})^\beta) \,.
\end{equation}
We call $\tilde{A}:=\rho^{\frac{1}{2}+\frac{1}{2\beta}}\sigma^{-1}\rho^{\frac{1}{2}+\frac{1}{2\beta}}$. Note that $A$ and $\tilde{A}$ share the same eigenvalues. Therefore to prove that the function is monotonically increasing, we need to show that
\begin{equation}
\Trm(\tilde{A}^\beta \log{\tilde{A}^\beta})-\Trm(\tilde{A}^\beta)\log{\Trm(\tilde{A}^\beta)} - \Trm((\log{\rho})\tilde{A}^\beta) \geq 0 \,.
\end{equation}
The latter condition is just a consequence of the positivity of the relative entropy. Indeed, 
\begin{align}
\Trm(\tilde{A}^\beta\log{\rho}) &= \Trm(\tilde{A}^\beta)\Trm\left(\frac{\tilde{A}^\beta}{\Trm(\tilde{A}^\beta)}\log{\rho}\right) \\
&\leq \Trm(\tilde{A}^\beta) \Trm\left(\frac{\tilde{A}^\beta}{\Trm(\tilde{A}^\beta)}\log{\frac{\tilde{A}^\beta}{\Trm(\tilde{A}^\beta)}}\right) \\
& = \Trm(\tilde{A}^\beta \log{\tilde{A}^\beta})-\Trm(\tilde{A}^\beta)\log{\Trm(\tilde{A}^\beta)} \,.
\end{align}
The case $\alpha \mapsto D_{\alpha,1-\alpha}(\rho \| \sigma)$ is similar. We denote $\beta:=1-\alpha$ and  $A:=\rho^{-\frac{1}{2}+\frac{1}{2\beta}}\sigma \rho^{-\frac{1}{2}+\frac{1}{2\beta}}$. If $\rho \perp \sigma$, then the theorem already follows. If instead the two states are not orthogonal, $\alpha \mapsto D_{\alpha,1-\alpha}(\rho \| \sigma)$ is differentiable on $(0,1)$, and 
\begin{align}
\frac{\d}{\d \beta} D_{\alpha,1-\alpha}(\rho \| \sigma) = -\frac{1}{\beta^2 \Trm(A^\beta)}\left(\Trm(A^\beta)\log{\Trm(A^\beta)} - \Trm(A^\beta \log{A^\beta}) + \Trm(A^\beta\log{\rho}) \right) \,,
\end{align}
which again is positive due to the positivity of the relative entropy.
\end{proof}
The above result immediately implies that  $\lim_{\alpha \rightarrow 0}\mathfrak{D}_{\alpha,1-\alpha}(\rho \| \sigma) $ $=\min_{\sigma \in \mathcal{F}} D_{\min}(\rho \| \sigma)$. By monotonicity of the sandwiched R\'enyi divergence, we also have that $\lim_{\alpha \rightarrow 1^-}\mathfrak{D}_{\alpha,\alpha}(\rho \| \sigma) $ $= \min_{\sigma \in \mathcal{F}} D(\rho \| \sigma)$.  As we prove below, the result for $\alpha \rightarrow \infty$ is a consequence of the minimax Theorem~\cite[Lemma II.1]{mosonyi2023some}. We note that the latter results hold for any resource theory.

\begin{corollary}
We have that $\lim_{\alpha \rightarrow \infty }\mathfrak{D}_{\alpha,\alpha-1}(\rho) = \mathfrak{D}_{\max}(\rho)$ . 
\end{corollary}

\begin{proof}
We recall that the function $(\rho, \sigma) \rightarrow D_{\alpha,z}(\rho \| \sigma)$ is lower semicontinuous for all $(\alpha,z)\in \mathcal{D}$~\cite{rubboli2024new}.
Moreover, lower semicontinuity implies that
\begin{align}
\lim \limits_{\alpha \rightarrow \infty}\min \limits_{\sigma \in \mathcal{F}}D_{\alpha,\alpha-1}(\rho \| \sigma) &= \sup \limits_{\alpha > 2} \inf \limits_{\sigma \in \mathcal{F}}D_{\alpha,\alpha-1}(\rho \| \sigma)\\
&= \inf \limits_{\sigma \in \mathcal{F}}   \sup \limits_{\alpha > 2}D_{\alpha,\alpha-1}(\rho \| \sigma) \\
&= \min \limits_{\sigma \in \mathcal{F}}  \lim \limits_{\alpha \rightarrow \infty}D_{\alpha,\alpha-1}(\rho \| \sigma) \,,
\end{align}
i.e., we can exchange the limit with the minimum. Here, $\mathcal{F}$ is the free set of the resource theory. 
The first equality follows from the monotonicity in $\alpha$ established in Lemma~\ref{monotonicity}. Hence we can replace the limit with the supremum over $\alpha>2$. The second equality follows from the minimax theorem in~\cite[Lemma II.1]{mosonyi2023some} since the $D_{\alpha,z}$ are lower semicontinuous in the second argument in the range where they satisfy the data-processing inequality~\cite[Lemma 18]{rubboli2024new}. Moreover, by the minimax theorem, in the last equality, the infimum over the free states can be replaced with the minimum.
\end{proof}

\end{document}